\documentclass[11pt,letterpaper]{article}
\usepackage[utf8]{inputenc}
\usepackage[pdftex]{graphicx}
\usepackage{graphicx}
\usepackage{epstopdf,subcaption}
\usepackage{amssymb,amsmath,amsthm,cite}
\usepackage{algorithm,algcompatible}
\usepackage[margin=1in]{geometry}
\usepackage{color}

\newtheorem{theorem}{Theorem}[section]

\newtheorem{proposition}[theorem]{Proposition}

\theoremstyle{definition}

\theoremstyle{remark}

\newcommand{\R}{\mathbb{R}}

\newcommand{\ba}{\mathbf{a}}
\newcommand{\bb}{\mathbf{b}}

\newcommand{\bd}{\mathbf{d}}
\newcommand{\be}{\mathbf{e}}

\newcommand{\bu}{\mathbf{u}}
\newcommand{\bv}{\mathbf{v}}
\newcommand{\bx}{\mathbf{x}}
\newcommand{\by}{\mathbf{y}}
\newcommand{\bz}{\mathbf{z}}

\newcommand{\grp}{\mathfrak{G}}

\title{Group Invariant Dictionary Learning}

\author{Yong Sheng Soh}
\date{Department of Mathematics, National University of Singapore\\ 10 Lower Kent Ridge Road, Singapore 119076 \\[2ex]%
Institute of High Performance Computing\\ 1 Fusionopolis Way, \# 16-16 Connexis, Singapore 138632 \\[2ex]%
15 July, 2020, revised 5 June, 2021
}

\begin{document}

\maketitle

\begin{abstract}
The dictionary learning problem concerns the task of representing data as sparse linear sums drawn from a smaller collection of basic building blocks.  In application domains where such techniques are deployed, we frequently encounter datasets where some form of symmetry or invariance is present.  Motivated by this observation, we develop a framework for learning dictionaries for data under the constraint that the collection of basic building blocks remains invariant under such symmetries.  Our procedure for learning such dictionaries relies on representing the symmetry as the action of a matrix group acting on the data, and subsequently introducing a convex penalty function so as to induce sparsity with respect to the collection of matrix group elements.  Our framework specializes to the convolutional dictionary learning problem when we consider integer shifts.  Using properties of positive semidefinite Hermitian Toeplitz matrices, we develop an extension that learns dictionaries that are invariant under continuous shifts.  Our numerical experiments on synthetic data and ECG data show that the incorporation of such symmetries as priors are most valuable when the dataset has few data-points, or when the full range of symmetries is inadequately expressed in the dataset.
\end{abstract}

\emph{Keywords}: sparse coding, equivariance, atomic norms, circulant matrices, orbitopes.
\section{Introduction}

The dictionary learning problem (also known as sparse coding in the literature) concerns the task of representing data as sparse linear sums of a smaller collection of basic building blocks:  Given a dataset $\{ \by^{(i)} \}_{i=1}^{n} \subset \R^d$, compute a collection vectors $\{ \ba_j \}_{j=1}^{q} \subset \R^{d}$ so that
\begin{equation} \label{eq:dl}
\by^{(i)} \approx \sum_{j=1}^{q} x^{(i)}_j \ba_j, ~ \text{s.t.}~ \bx^{(i)} = (x_1^{(i)},\ldots,x_q^{(i)})^{\intercal} \text{ is sparse }\forall i.
\end{equation}

The dictionary learning task is motivated by the prevalence of sparse representations in a wide range of data processing applications.  Sparse representations form the basis of numerous procedures for storage, compression, as well as communication of data.  In addition, numerous computational procedures for downstream processing tasks such as denoising and the imputation of missing entries heavily rely on data admitting sparse representations for its success.

A fundamental ingredient for applying these methods is that we identify a suitable transformation -- frequently referred to as a basis -- under which our dataset of interest admits sparse representations.  The traditional process of identifying such transformations relies on extensive knowledge about the data.  For instance, we frequently deploy the collection of wavelets transforms and discrete cosine transforms in image processing applications because of well-known properties about natural images.  The dictionary learning procedure may be viewed as a data-driven alternative in which an appropriate choice of basis is \emph{learned} directly from data \cite{OlsFie:96,OlsFie:97,AEB:06,MBP:14,Ela:10}.  As the resulting basis is specifically tuned to the dataset, it enjoys better performance compared to choices of bases specified using prior knowledge in many instances \cite{ElaAha:06,Ela:10}.  More importantly, dictionary learning is useful in instances where one lacks the appropriate domain specific knowledge to identify a basis -- one simply applies dictionary learning techniques to learn a suitable choice of basis.

\subsection{Group Invariant Dictionaries}

Symmetries and invariances occur in a wide range of scientific and engineering domains.  For instance, in image processing applications where the data takes the form of image patches segmented from larger natural images, one might expect that these patches possess some form of translation or rotation invariance.  In time series analyses where the data takes the form of short time series segmented from longer time series, one might expect that these time series possess some form of shift invariance occurring across time.  In processing data over graphs, one might consider invariance with respect to re-labelling of nodes, while in tomographic applications, one might wish to incorporate invariances with respect to rigid rotations.  In view of the prevalence of such symmetries arising in applications, it is natural to consider learning dictionaries that also respect such symmetries.

There are several concrete advantages to learning such dictionaries.  First, the incorporation of symmetries allows us to identify multiple basis elements as being equivalent up to an appropriate transformation.  This identification allows us to reduce the degrees of freedom that are involved in the dictionary learning task, and subsequently learn bases with greater statistical relevance.  Second, the incorporation of such symmetries prevents the learned dictionary from introducing unintended biases; for instance, in learning dictionaries for image patches, we may prefer dictionaries that do not favor upright orientations of certain dictionary elements.

\textbf{Reducing sample complexity by incorporating invariances.}  In the following, we make the case for learning dictionaries that incorporate the appropriate invariances via a numerical experiment on synthetically generated data.  Our description is brief, and we defer further experimental details to Section \ref{sec:numexp}.  Given a vector $(d_1,\ldots,d_q)^{\intercal}$, we define an \emph{integer shift} as any vector of the form $(d_k,d_{k+1},\ldots,d_q,d_1,\ldots,d_{k-1})^{\intercal}$ for some $k$, $0\leq k \leq q-1$.  We consider learning a dictionary from a dataset $\{ \by^{(i)} \}_{i=1}^{1000} \subset \R^{30}$ that is generated from a dictionary that possesses integer shift invariance -- that is, each data-point $\by^{(i)}$ is expressible as the linear sum of a small number of integer shifts of a collection of vectors.  In Figure \ref{fig:ComparisonRecovery_Intro}, we compare the performance of two different dictionary learning algorithms -- the first incorporates integer shift invariances as a prior, and the second is Regular Dictionary Learning (DL) which does not incorporate such a prior.  We repeat the experimental set-up over $10$ different random initializations of both algorithms, and we compare the distance between the ground truth dictionary and each iterate of both algorithms.  We observe that our framework converges to the ground truth dictionary in about $30$ iterations whereas the iterates from Regular DL do not converge to the ground truth, even after $100$ iterates.

\begin{figure}
	\centering
	\includegraphics[width=0.6\textwidth]{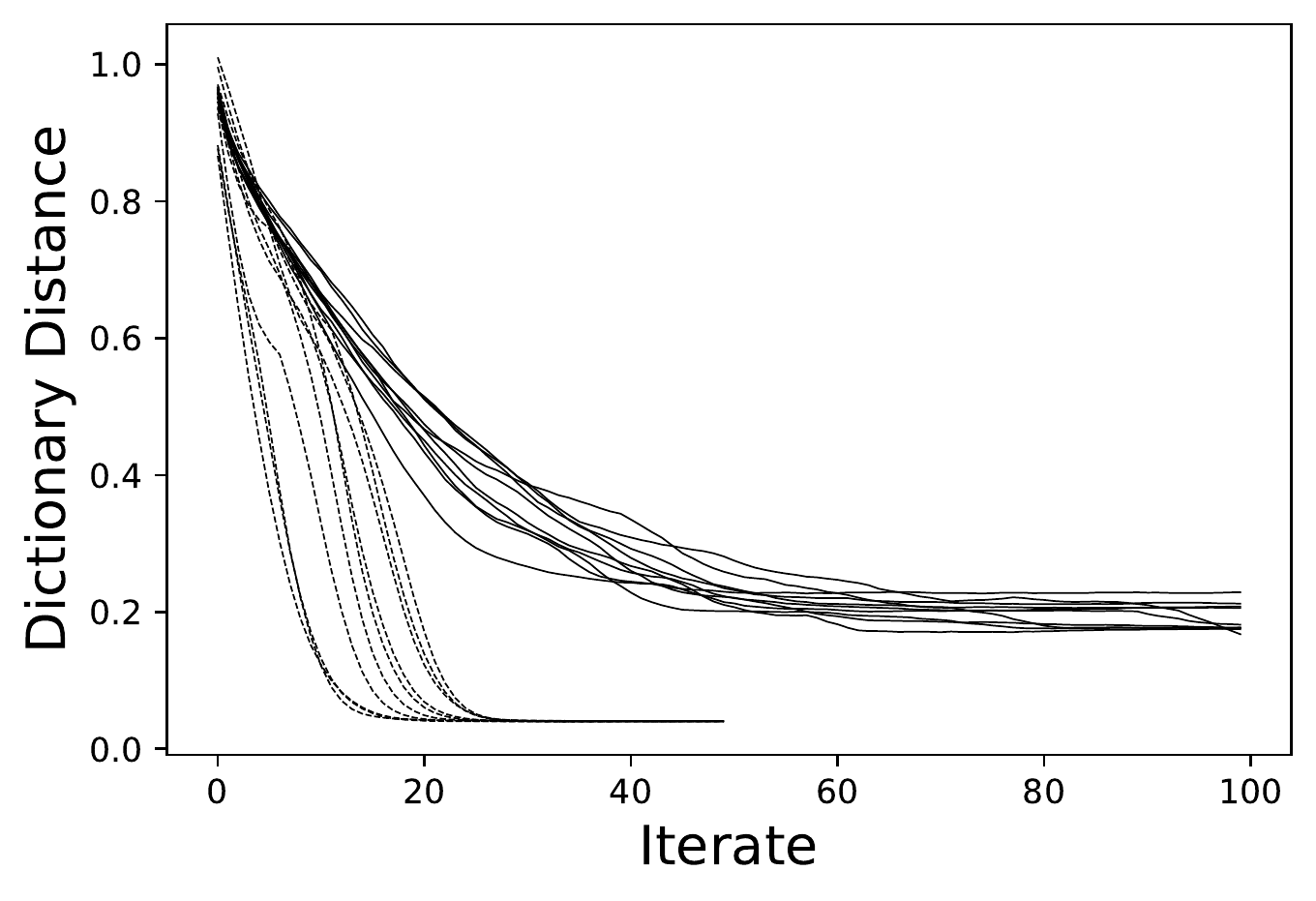}
	\caption{Comparison of learning a shift invariant dictionary using an algorithm that incorporates shift invariance as a prior (dashed lines) with an algorithm that does not (solid lines).}
	\label{fig:ComparisonRecovery_Intro}
\end{figure}

\subsection{Prior and Related Works}

\noindent \textbf{Convolutional Dictionary Learning.}  Our framework is motivated by a line of work that learns dictionaries possessing \emph{integer shift} invariance.  Concretely, let $\mathcal{D} \subset \R^{q}$ be a collection of vectors.  We say that $\mathcal{D}$ is \emph{integer shift invariant} if $\ba \in \mathcal{D}$ implies that all integer shifts of $\ba$ are also in $\mathcal{D}$.  The goal in Convolutional Dictionary Learning (DL) is to learn dictionaries that are integer shift invariant \cite{Wohlberg:15,LewSej:17,PRSE:17,JVLG:06,BluDav:06,PABD:06,GRKN:07,Rusu:20,ZSE:19,GW:18,LGWY:18}.    

The Convolutional DL problem admits a more compact description by means of convolutions (as its name suggests):    Given vectors $\bu,\bv \in \R^d$, the convolution $ \bu * \bv$ is the $d$-dimensional vector whose $i$-th coordinate is the sum $\sum_{k=1}^{q} u_{k} v_{i-k}$.  The Convolutional DL problem can thus be described as one of learning a collection of vectors $\{ \ba_j \}_{j=1}^{q}$ such that data is well approximated as linear sums of convolutions of these $\ba_j$'s with a corresponding collection of vectors of equal dimension:
\begin{equation} \label{eq:cdl}
\by^{(i)} \approx \sum_{j=1}^{q} \ba_j \ast \bx^{(i)}_j , ~\text{s.t.}~ \bx^{(i)}_j \in \R^{d}, \bx^{(i)}_j  ~ \text{is sparse for all }i,j.
\end{equation}
Convolutional DL techniques are used in a range of application domains where data exhibit shift invariances, with image processing and audio processing being prominent examples \cite{JVLG:06,BluDav:06,PABD:06,GRKN:07}. \\

\noindent \textbf{Group Convolutional Neural Networks.}  Broadly speaking, convolutional neural networks can be viewed as instances of neural networks that incorporate an appropriate form of shift invariances.  Such techniques have been empirically observed to be significantly more powerful (and perhaps considered state-of-the-art) than traditional neural network architectures that do not incorporate such priors in tasks such as image classification \cite{KSH:12}.  

There is a body of work that seeks to extend the ideas of convolutional neural networks to more general symmetries \cite{WGTB:17,GD:14,CW:16,KT:18}.  The basic ingredient in these networks is that the output from applying transformations (such as shifts or rotations) in the input layer should yield the same outcome had we only apply the same set of transformations on the output layer.  Such a property is known as \emph{equivariance}, and is key for generalizing the structure of convolutional neural networks to accommodate more general symmetries.  

The key difference between this body of work and ours is that these works do \emph{not} seek sparse representations whereas sparse representations are central to our set-up, which subsequently necessitates the development of appropriate penalty functions to achieve our goal.  In addition, while some of these works such as \cite{CW:16,KT:18} prescribe conceptual frameworks for incorporating group equivariant structure within neural networks, there remains a significant gap in practically implementing these methods especially in settings where the group is continuous \cite{WGTB:17}.  In contrast, one of the key contributions of our framework is to overcome certain computational difficulties that arise precisely because of continuous symmetries.  In Section \ref{sec:conc}, we discuss future directions stemming from our work in the context of neural networks.  \\

\noindent \textbf{Semidefinite Programming-Representable Regularizers.}  In Section \ref{sec:example_ctsshift}, we describe an approach for learning dictionaries that are continuously shift invariant.  As we later show, a key ingredient is to express the regularizers associated to such dictionaries via semidefinite programming (SDP).  A prior work that is conceptually related is \cite{SC:19}, which proposes an extension of dictionary learning to that of learning \emph{infinite dictionaries} that are expressible via SDPs.  The key difference between our work and \cite{SC:19} lies in the role that SDP descriptions play -- in \cite{SC:19}, SDP descriptions offer a framework for describing infinite collection of basic building blocks in a tractable fashion -- notably, the framework does \emph{not} incorporate any form of symmetry; in Section \ref{sec:example_ctsshift}, SDP descriptions arises because of the type of symmetry we wish to incorporate in the dictionary.

\subsection{Our Contributions}

In this paper, we introduce an algorithmic framework for learning dictionaries that are invariant under more general symmetries.  In particular, our framework generalizes prior methods for Convolutional DL.

The key technical difficulty in generalizing these prior works lies in identifying a suitable parameterization of a dictionary that is group invariant.  To provide some context, prior methods for dictionary learning operate on the basis of performing updates in alternating directions in which (i) we fix a linear map $D \in \R^{d\times q}$ representing dictionary elements and, given data vectors $\{\by^{(i)}\}_{i=1}^{n} \in \R^d$, compute sparse vector $\{\bx^{(i)}\}_{i=1}^{n} \in \R^{q}$ such that $\by^{(i)} \approx D \bx^{(i)}$, and (ii) fix the sparse vectors $\{\bx^{(i)}\}_{i=1}^{n} \in \R^{q}$ and update a linear map $D$ so that $\by^{(i)} \approx D \bx^{(i)}$. These methods require us to express the entire dictionary $\mathcal{D}$ explicity, and hence are no longer feasible if, for instance, the $\mathcal{D}$ contains a continuum of elements.

To address such difficulties, we focus on symmetries that are expressible as a matrix group action.  More specifically, our framework requires every dictionary element to be expressible as the \emph{orbit} of some \emph{matrix group} acting on a collection of \emph{generators}:
\begin{equation} \label{eq:dict_orbit}
\mathcal{D} = \left\{  G \, \ba \,:\, G \in \grp, \ba \in \{ \ba_j \}_{j=1}^{q} \right\}.
\end{equation}
The matrix group $\grp$ expresses the symmetry, and the collection of generators $\{ \ba_j \}_{j=1}^{q}$ are the basic (or canonical) atoms from which we describe the entire dictionary. 

In Section \ref{sec:framework}, we introduce our framework for learning group invariant dictionaries based on solving the following minimization instance:
\begin{equation*}
\underset{\substack{ \{Z_j^{(i)}\}_{i,j=1}^{q,n} \subset \R^{d\times d}\\ \{\ba_j\}_{j=1}^{q} \subset \R^{d} }}{\arg \min}  ~ \sum_{i=1}^{n} ( \frac{1}{2} \| \by^{(i)} - \sum_{j} Z_j^{(i)}  \ba_j \|_2^2 + \lambda \sum_j \|Z_j^{(i)}\|_{\grp}). \end{equation*}
Here,
\begin{equation*}
\| Z \|_{\grp} ~ := ~ \inf \{ t : t Z \in \mathrm{conv}(\grp) \}.
\end{equation*}
is a penalty function that is useful for promoting succinct representations with respect to the collection $\grp$.  The penalty function is known as an \emph{atomic norm} in the literature \cite{CRPW:12}.

The key conceptual contribution in our framework that addresses the issues we raise is the following: Conventional wisdom tells us to decouple the learning task into the variables $D$ and $\bx^{(i)}$ -- the linear map $D$ contains the generators \emph{and} its transformed copies, while the vectors $\bx^{(i)}$ only contains the sparse coding information.  Our framework proposes decoupling the learning task into the variables $\ba_j$ and $Z^{(i)}_j$ -- the vectors $\ba_j$ represent the generators but \emph{not} the transformed copies, while the matrices $Z^{(i)}_j$ combine sparse coding information with group transformation information.  As we show in Section \ref{sec:framework}, the dictionary learning task reduces to one of finding a suitable collection of generators $\{ \ba_j \}_{j=1}^{q}$.  This reduction is crucial because it allows us to learn infinite dictionaries so long as they are finitely parameterized.

\subsection{Notation}

In the remainder of this paper, we adopt the notational convention whereby $a$ denotes a scalar, $\ba$ denotes a vector, $A$ denotes a matrix, and $\mathcal{A}$ denotes a set or a collection.  Given a set $\mathcal{C}$, we denote the induced norm $ \| \bx \|_{\mathcal{C}} :=  \inf \{ t : t \, \bx \in \mathrm{conv}(\mathcal{C}) \}$.  Then $\| \bx \|_{\mathcal{C}} < +\infty$ if $\bx \in \mathrm{Span}(\mathcal{C})$, and we adopt the convention $\| \bx \|_{\mathcal{C}} = +\infty$ if $\bx \notin \mathrm{Span}(\mathcal{C})$, as is standard in convex analysis.  Subsequently, minimizing objectives that incorporate $\| \cdot \|_{\mathcal{C}}$ as a penalty necessarily enforces the solution to reside in $ \mathrm{Span}(\mathcal{C})$.  Finally, an superscript asterisk $^*$ (but not $^\star$) denotes the complex conjugate, a regular asterisk $*$ denotes convolution, $^\intercal$ denotes the regular transpose, while $^\dagger$ denotes the Hermitian transpose.
\section{Framework} \label{sec:framework}

We define a \emph{dictionary} $\mathcal{D} \subset \R^d$ to be a collection of vectors that is possibly infinite.  The dictionary represents the basic building blocks from which we describe our dataset of interest.  As such, the dictionary learning problem can be described as one of computing an appropriate dictionary for a given dataset so that each data-point is well approximated as the linear sum of few dictionary elements.

In the following, we describe the key ingredients necessary to apply our framework.
\begin{enumerate}
\item[(A1)] \emph{Matrix group representation.} First, we require the invariance to be expressible as the \emph{linear action} of a matrix group $\grp$  acting on data.  More precisely, given data $\bx \in \mathbb{R}^{d}$, we assume that \emph{every} transformed version of $\bx$ is expressible as
\begin{equation*} \label{eq:assumption_linear}
G \bx \quad \text{ where } \quad G \in \grp,  ~ G \in \mathrm{GL}(\mathbb{R},d). 
\end{equation*}
We say that a dictionary $\mathcal{D}$ is \emph{invariant} with respect to $\grp$ if $\ba \in \mathcal{D}$ implies that $G \ba \in \mathcal{D}$ for all $G \in \grp$ and all $\ba \in \mathcal{D}$.

\item[(A2)] \emph{Finite generation.}  Second, our framework is only applicable to learning dictionaries that are \emph{finitely generated}.  More precisely, we require $\mathcal{D}$ to be expressible as the action of $\grp$ acting on a finite collection of \emph{generators} $\{ \ba_1,\ldots,\ba_q \}$:
\begin{equation*} \label{eq:assumption_finitegeneration}
\mathcal{D} = \{ G \, \ba : \ba \in \mathcal{A}, \, G \in \grp \}, \mathcal{A} = \{ \ba_1,\ldots,\ba_q \}.
\end{equation*}
From a practical consideration, the finite generation stipulation is by no means restrictive.  The reason we emphasize finite generation in our set-up is because it translates to a finite parameterization of $\mathcal{D}$ even if the group $\grp$ is infinite.

\item[(A3)] \emph{Origin symmetry.}  Third, we require the matrix group $\grp$ to contain the negative identity matrix $-I$:
\begin{equation*} \label{eq:assumption_reflection}
-I \in \grp .
\end{equation*}
Note that it makes practical sense to include such an assumption -- it simply reflects the fact that if $\bd$ is a dictionary element, then one would reasonably expect that its negation $-\bd$ is also contained in $\mathcal{D}$. 

\item[(A4)] \emph{Tractable descriptions of $\mathrm{conv}(\grp)$.}  Our fourth ingredient requires to provide tractable descriptions of the set $\mathrm{conv}(\grp)$.  At a conceptual level, our stipulation is equivalent to being able to optimize over the set $\mathrm{conv}(\grp)$ tractably, which we require as a sub-routine in our procedure.
\end{enumerate}

\subsection{Succinct Representations with respect to a Dictionary}  

Our algorithm for learning a group invariant dictionary requires us to perform the following task as a sub-routine:  
\begin{quotation}
Given a (finitely generated) dictionary $\mathcal{D} \subset \R^{d}$ and a vector $\by \in \R^{d}$, compute an approximation of $\tilde{\by} \approx \by$ such that $\tilde{\by}$ is expressible as the linear sum of few elements from $\mathcal{D}$.
\end{quotation}
    
The analogue of the above procedure when specialized to Regular Dictionary Learning entails the computational task of recovering as sparse vector from affine measurements:  Given a linear map $D \in \R^{d \times q}$ and a vector $\by \in \R^d$, compute a suitably sparse vector $\bx \in \R^{q}$ so that $\by \approx \tilde{\by} = D \bx$.  Here, the columns of the linear map $D$ are the dictionary elements, and the constraint that $\bx$ is sparse is equivalent to the requirement that we use few dictionary elements in the approximation of $\by$.  The na\"ive approach to recovering the sparse vector $\bx$ is to employ a combinatorial search, which is infeasible for problem instances in moderate to large dimensions.  However, due to its importance in a wide range of statistical and signal processing tasks, numerous procedures that work well in practice and provably work in certain instances have been developed.  One such class of methods is based on a convex relaxation in which we estimate $\bx$ by minimizing a least squares loss augmented with a L1-norm \cite{Don:06,Don:06b,CRT:06}:
\begin{equation} \label{eq:compressedsensing}
\underset{\bx}{\mathrm{argmin}} ~ \frac{1}{2} \| \by - D \bx \|_2^2 + \lambda \|\bx\|_1.
\end{equation}
Here, $\lambda$ denotes a (positive) regularization parameter.  Such methods are particularly powerful because these are based on solutions of a tractable optimization instance, and are provably effective at finding the sparsest solutions to the problem \cite{Don:06,Don:06b,CRT:06}.  In the next section, we introduce the appropriate generalization of \eqref{eq:compressedsensing} for our set-up.

\subsection{Succinct Representations via Atomic Norms} \label{sec:atomic_norms}

Let $\mathcal{C} \subset \R^{d}$ be a compact set.  We say that a vector $\by \in \R^d$ admits a succinct representation with respect to $\mathcal{C}$ if it is expressible as the linear sum of a small number of elements from $\mathcal{C}$:
\begin{equation*}
\by = \sum_{i \in \mathcal{I}} c_{i} \ba_i \quad \text{where} \quad \ba_i \in \mathcal{C}, \quad \text{and} \quad |\mathcal{I}| \ll d.
\end{equation*}
We remark that the cardinality of the set $\mathcal{C}$ is permitted to be arbitrarily large; in particular, the set $\mathcal{C}$ may be infinite or uncountable.  This notion of succinct representations generalizes several notions of structured signals arising in applications.  For instance, the collection of structured objects when $\mathcal{C}  = \{ \pm \be_j : 1\leq j \leq q \}$ is specialized to the collection of signed standard basis vectors corresponds to sparse vectors.  Similarly, the collection of structured objects when $\mathcal{C} = \{ \bu \bv^{\prime} : \|\bu\|_2 = \|\bv\|_2 = 1 \}$ is specialized to the collection of rank-one matrices with unit Frobenius-norm corresponds to low-rank matrices.  We refer the interested reader to \cite{CRPW:12} for a more extensive list of examples.

The key property concerning objects that admit succinct representations with respect to some collection $\mathcal{C} \subset \R^{d}$ is that the \emph{atomic norm} induced by the convex hull of $\mathcal{C}$ is a convex penalty function that is effective at inducing structure as succinct representations with respect to $\mathcal{C}$.  More formally, we define the atomic norm with respect to $\mathcal{C}$ as the following function \cite{CRPW:12}:
\begin{equation} \label{eq:atomicnorm_defn}
\| \bx \|_{\mathcal{C}} = \inf \, \{  t : \bx \in t \cdot \mathrm{conv}(\mathcal{C}) \}.
\end{equation} 
The function $\| \cdot \|_{\mathcal{C}}$ is also known as the gauge function or the Minkowski functional defined with respect to $\mathrm{conv}(\mathcal{C})$.  The convexity of $\| \cdot \|_{\mathcal{C}}$ follows from the fact that $\mathrm{conv}(\mathcal{C})$ is convex.  In order for the function $\| \cdot \|_{\mathcal{C}}$ to define a true \emph{norm}, we also require the set $\mathrm{conv}(\mathcal{C})$ to be centrally symmetric -- that is, $\bx \in \mathrm{conv}(\mathcal{C})$ if and only if $-\bx \in \mathrm{conv}(\mathcal{C})$.  In the sequel, we take $\mathcal{C}$ to be the set $\mathcal{D}$, and hence $\| \cdot \|_{\mathcal{C}}$ defines a true norm whenever $-I \in \grp$ -- this is precisely Assumption (A3).

Under the additional assumption that the centroid of the set $\mathrm{conv}(\mathcal{C})$ is at the origin -- this is satisfied if $\mathcal{C}$ is centrally symmetric and compact -- then we have an alternative characterization of the atomic norm:
\begin{proposition} \label{thm:atomicnorm_sumcoeff}  Suppose the centroid of $\mathrm{conv}(\mathcal{C})$ is the origin.  Then
\begin{equation*}
\| \bx \|_{\mathcal{C}} = \inf \, \left\{  \sum c_j : \bx = \sum c_j \ba_j, \ba_j \in \mathcal{C}, c_j > 0 \right\}.
\end{equation*}
\end{proposition}

The atomic norm generalizes choices of convex penalty functions that are widely used to induce structure as sparse vectors or low-rank matrices.  Specifically, the atomic norm when specialized to $\mathcal{C} = \{ \be_j : 1\leq j \leq q \}$ being the collection of standard basis vectors recovers the L1-norm, and the atomic norm when specialized to $\mathcal{C} = \{ \bu \bv^{\prime} : \|\bu\|_2 = \|\bv\|_2 = 1 \}$ being the collection of rank-one matrices with unit Frobenius-norm recovers the matrix nuclear-norm (also known as the Schatten 1 norm).  In the following, we apply the atomic norm induced by the dictionary $\mathcal{D}$ to approximate a data vector sparsely with respect to $\mathcal{D}$.

\subsection{Atomic Norms for Group Invariant Dictionary Learning} \label{sec:atomicnorms_forgidl}

In the remainder of this section, let $\mathcal{D} = \{ G \, \ba : G \in \grp, \ba \in \mathcal{A} \}$,  $\mathcal{A} = \{ \ba_1,\ldots,\ba_q \}$, be a finitely generated dictionary.  
A prominent class of methods for representing an input sparsely with respect to some given dictionary $\mathcal{D}$ is to apply the \emph{proximal operator} with respect to the atomic norm induced by $\mathcal{D}$.  The use of such operators for performing denoising and obtaining sparse representations was initially studied in \cite{Donoho:95,DonJoh:94} -- more frequently referred to as soft-thresholding -- and later extended to general atomic norms in \cite{BTR:13}.  In our context, we obtain $\tilde{\by}$ as the solution of the following minimization instance:
\begin{equation} \label{eq:proxoperator}
\tilde{\by} \in \underset{\bz \in \R^{d}}{\arg \min} ~ \frac{1}{2} \| \by - \bz \|_2^2 + \lambda \| \bz \|_{\mathcal{D} }.
\end{equation}

However, the regularizer $\| \cdot \|_{\mathcal{D} }$ -- while defined abstractly in \eqref{eq:atomicnorm_defn} -- is not in a form that is evidently amenable to computation.  Our final ingredient is to provide an alternative characterization of $\| \cdot \|_{\mathcal{D} }$ in terms of the group $\grp$ and the generators $\mathcal{A}$.  Define the following atomic norm $\| \cdot  \|_{\grp}$ over the space of matrices $\R^{d\times d}$:
\begin{equation} \label{eq:grpatomicnorm}
\| Z \|_{\grp} = \inf \{ t : t Z \in \mathrm{conv}( \grp ) \}.
\end{equation}
Note that since Assumption (A3) guarantees that $\grp$ is centrally symmetric, the expression in \eqref{eq:grpatomicnorm} defines a norm. 
As a result of Proposition \ref{thm:atomicnorm_sumcoeff} as well as Assumptions (A1) and (A2), we obtain the following equivalent expression for the atomic norm $\| \cdot \|_{\mathcal{D} } $:
\begin{proposition} \label{thm:norm_equivexp}  
Given a matrix group $\grp$ and a finite collection of generators $\mathcal{A} = \{ \ba_1,\ldots,\ba_q\}$, let $\mathcal{D} = \{ G \, \ba : G \in \grp, \ba \in \mathcal{A} \}$ be the associated dictionary.  Suppose $-I \in \grp$.  Then 
\begin{equation*}
\| \bx \|_{\mathcal{D}} ~=~   \underset{Z_j \in \R^{d\times d}, 1\leq j \leq q}{\inf} ~\left\{ \sum_{j=1}^{q} \| Z_j \|_{\grp}  : \bx = \sum_{j=1}^{q} Z_j \ba_j \right\}.
\end{equation*}
\end{proposition}

\begin{proof}[Proof of Proposition \ref{thm:norm_equivexp}]
Let $\bx$ be arbitrary.  Let $\mathrm{LHS}:= \| \bx \|_{\mathcal{D}}$, and let $\mathrm{RHS}$ denote the optimal value of the right hand side expression.  Note that if $\bx$ is not in the span of $\mathcal{D}$, then the LHS is $+\infty$ by definition from \eqref{eq:atomicnorm_defn}, and the RHS is $+\infty$ by infeasibility.  As such, we may assume that $\bx \in \mathrm{Span}(\mathcal{D})$.  Fix $\epsilon > 0$.  Then we have $\bx = \sum_{k} c_k (G_k \bb_{k})$ where $c_k>0$, $G_k \in \grp$, $\bb_{k} \in \mathcal{A}$, and $\sum c_k < \mathrm{LHS} + \epsilon$.  Let $Z_j = \sum_{k} c_k \delta(\bb_{k},\ba_j) G_k$, where $\delta(\cdot,\cdot)$ is the Kronecker delta function that evaluates to one if and only if the two arguments are equal.  Then $\sum_{j=1}^{q} Z_j \ba_j = \sum_{j=1}^{q} ((\sum_{k} c_k \delta(\bb_{k},\ba_j) G_k) \ba_j) = \sum_{k} ((\sum_{j=1}^{q}  c_k \delta(\bb_{k},\ba_j) G_k) \ba_j)  = \sum_{k} c_k ( G_k \bb_k) = \bx$; i.e., the matrices $Z_j$ form a feasible solution to the convex program on the right hand side.  Based on Proposition \ref{thm:atomicnorm_sumcoeff}, it follows that $\|Z_j\| \leq \sum_{k} c_k \delta(\bb_{k},\ba_j)$.  Summing across $j$, we have $\sum_{j} \|Z_j\| \leq \sum_{j} \sum_{k} c_k \delta(\bb_{k},\ba_j) = \sum_{k} (\sum_{j}  c_k \delta(\bb_{k},\ba_j) ) = \sum_{k} c_k < \mathrm{LHS} + \epsilon$.  Finally, we take $\epsilon \rightarrow 0$ to conclude that $\mathrm{RHS} \leq \mathrm{LHS}$.  

An essentially similar set of arguments applied in the opposite direction gives us the reverse inequality $\mathrm{LHS} \geq \mathrm{RHS}$, from which we conclude the proof.
\end{proof}

Based on characterization result in Proposition \ref{thm:norm_equivexp}, our procedure for finding an approximation of $\by$ as a linear sum of few elements from $\mathcal{D}$ is to compute the optimal solution for the following minimization instance:
\begin{equation} \label{eq:atomicnorm_softthresholding}
	\underset{ \{Z_j\}_{j=1}^{q} \subset \R^{d\times d}}{\arg \min} ~~ \frac{1}{2} \| \by - \sum_{j} Z_j \ba_j \|_2^2 + \lambda \left( \sum_j \|Z_j\|_{\grp} \right).
\end{equation}

Note that our ability to minimize \eqref{eq:atomicnorm_softthresholding} tractably relies on us having tractable descriptions of the convex hull of $\grp$ -- this is precisely what Assumption (A4) stipulates.  Very frequently, tractable descriptions of $\mathrm{conv}(\grp)$ lead to tractable numerical procedures for minimizing \eqref{eq:atomicnorm_softthresholding}.

\subsection{Learning Group Invariant Dictionaries via Matrix Factorization}

We put the pieces together, and formally state our framework for learning group invariant dictionaries.  Let $\{ \by^{(i)} \}_{i=1}^{n} \subset \R^d$ denote the dataset of interest.  Suppose we wish to learn a dictionary that is invariant under the action of the matrix group $\grp$.  We do so by solving the following minimization instance:
\begin{equation} \label{eq:GIDL}
\begin{aligned}
\underset{\substack{ \{Z_j^{(i)}\}_{i,j=1}^{q,n} \subset \R^{d\times d}\\ \{\ba_j\}_{j=1}^{q} \subset \R^{d} }}{\arg \min} & ~ \sum_{i=1}^{n} ( \frac{1}{2} \| \by^{(i)} - \sum_{j} Z_j^{(i)}  \ba_j \|_2^2 + \lambda \sum_j \|Z_j^{(i)}\|_{\grp}). \\
\mathrm{s.t.} & \qquad \| \ba_j \|_2 = 1 \text{ for all } j
\end{aligned}
\end{equation}
Here, we impose the additional constraint $\| \ba_j \|_2 = 1$ to ensure that the resulting dictionary elements are well conditioned.

Suppose we let $\{ \hat{\ba}_{j} \}_{j=1}^{q}$ be the resulting optimal set of generators obtained from \eqref{eq:GIDL}.  Then the learned dictionary is $\mathcal{D} = \{ G \, \hat{\ba}_{j} : 1\leq j \leq q, G \in \grp \}$.

\subsection{Intermission}

We make a series of useful remarks concerning our framework.

\begin{enumerate}
\item From a conceptual perspective, Proposition \ref{thm:norm_equivexp} decouples the atomic norm $\|\cdot\|_{\mathcal{D}}$ into two components:  The first of which is the penalty function $\| \cdot \|_{\grp}$, and it depends purely on the symmetry group $\grp$.  The second of which is the affine equality $\bx = \sum_{j=1}^{q} Z_j \ba_j$, and it depends purely on the generators.  More importantly, the decoupling is essential to developing a computational algorithm for learning group invariant dictionaries -- it naturally leads to a strategy based on minimizing with respect to each variable, and which generalizes existing dictionary learning frameworks.

\item Since the penalty term in the objective \eqref{eq:atomicnorm_softthresholding} is an atomic norm, we know that the optimal solution typically admits succinct representations with respect to $\grp$.  However, we do not -- generally speaking -- know how the optimal solution decomposes succinctly as elements in $\grp$.  

\item In relation to our earlier point, we emphasize that our framework relies on the following somewhat surprising observation:
\begin{quotation}
For dictionary learning, we do \emph{not} require the explicit decomposition of a data-point into its atomic constituents (with respect to a dictionary estimate) in order to learn a dictionary. 
\end{quotation}
We make this point clearer with Regular Dictionary Learning:  The first step of any iterative procedure is to compute a sparse vector $\bx^{(i)}$ such that $\by^{(i)} \approx D \bx^{(i)}$ for every data vector $\by^{(i)}$.  The second step is to update the dictionary estimate using the obtained vectors $\{\bx^{(i)} \}_{i=1}^{n}$.  If we perform the second step via a least squares minimization \cite{EAH:99} or a gradient descent, then there is no instance in which one has to refer to the exact location of the non-zero entries.\footnote{To be absolutely clear, there are certain classes of Regular Dictionary Learning algorithms such as the K-SVD \cite{AEB:06} and the ITKM \cite{Sch:16} that do depend on knowing the support; that is, the location of the non-zero entries of $\bx^{(i)}$.  Roughly speaking, these algorithms update the dictionary elements \emph{sequentially}, and as such require knowledge of the location of non-zero entries.}

\item Conversely, there are certain applications such as classification tasks where the explicit decomposition of a signal into its constituents is required.  The problem of obtaining the sparsest decomposition of a signal is computationally difficult, though a wide range of approximate techniques are available in the literature \cite{CDS:98}.

\item We emphasize that, given a generic group $\grp$, there is no general procedure for providing descriptions of $\mathrm{conv}(\grp)$ (and especially tractable ones) that are amenable to optimization.  Nevertheless, many interesting examples do admit tractable descriptions, and we describe some of these in Section \ref{sec:examples}.  We have strong reasons to believe that the question of providing tractable descriptions of $\mathrm{conv}(\grp)$ for generic $\grp$ is likely to be hard in general given that the broader question of providing tractable descriptions of convex sets is widely known to be difficult.  In fact, the question of providing conic programming representations of matrix groups remains an active research area (see Section \ref{sec:conc}).  It is for these reasons we state Assumption (A4) as it is.  

\end{enumerate}
\section{Algorithm} \label{sec:algorithm}

In this section, we describe our algorithm for learning group invariant dictionaries.  Our procedure is based on minimizing the objective \eqref{eq:GIDL} in alternating directions, and it relies on the observation that the objective \eqref{eq:GIDL} -- when keeping the variables $\{ \ba_j \}$ or the variables $\{Z_j^{(i)}\}$ fixed -- leads to a convex program.  Our algorithm also generalizes prior methods for regular dictionary learning and convolutional dictionary learning.  We summarize the description of our procedure in Algorithm \ref{alg:gdl_am}.

\begin{algorithm}[h] 
	\caption{Alternating Minimization-based Algorithm for Learning Group Invariant Dictionaries}
	\label{alg:gdl_am}
	\textbf{Input}: Initialization $\{\ba_j\}_{j=1}^{q}$, Data $\{\by^{(i)}\}_{i=1}^{n}$  \\
	\textbf{Require}: Normalized dictionary generators $\{\ba_{j} \}_{j=1}^{q} $ \\
	\textbf{Algorithm}: Repeat until convergence \\
	\textbf{1.}[Fix $\ba_i$, update $Z_j^{(i)}$] Update $Z_1^{(i)}, \ldots, Z_q^{(i)}$ as solutions to the following convex program
	\begin{equation*}
	\underset{Z_1,\ldots,Z_q}{\arg \min} ~ \frac{1}{2} \| \by^{(i)} - \sum_{j} Z_j^{(i)} \ba_j \|_2^2 + \lambda \cdot \sum_j \|Z_j^{(i)}\|_{\grp}, 1 \leq i \leq n.
	\end{equation*}
	\textbf{2.}[Fix $Z_j^{(i)}$, update  $\ba_j$] Solve the least squares problem 
	\begin{equation*}
	\left(\ba_1,\ldots,\ba_q \right) \leftarrow \underset{\ba_1,\ldots,\ba_q}{\arg \min} ~ \sum_{i=1}^{n} \| \by^{(i)} - \sum_{j} Z_j^{(i)}  \ba_j \|_2^2.
	\end{equation*} 
	\textbf{3.}[Normalize] $\ba_j \leftarrow \ba_j / \|\ba_j\|_2$
\end{algorithm}

\textbf{Sparse representations via Atomic Norm Regularization.}  The first step in each iteration is an update step in which we keep the variables $\{ \ba_j \}_{j=1}^{q}$ fixed, and we minimize the objective \eqref{eq:GIDL} with respect to the variables $\{Z_j^{(i)} \}_{j=1,i=1}^{q,n}$.  This entails updating $Z_1^{(i)}, \ldots, Z_q^{(i)} \in \R^{q\times q} $ as solutions of the following convex program:
\begin{equation} \label{eq:atomicnorm_sparsecoding}
\underset{Z_1,\ldots,Z_q}{\arg \min} ~ \frac{1}{2} \| \by^{(i)} - \sum_{j} Z_j \ba_j \|_2^2 +  \lambda \cdot \sum_j \|Z_j\|_{\grp}, ~ 1 \leq i \leq n.
\end{equation}

While the above update step can be tractably solved in many cases, practical implementations of methods for solving \eqref{eq:atomicnorm_sparsecoding} can be fairly expensive when the data dimension becomes moderately large.  We suggest some possible mitigation measures in the following:

\begin{enumerate}
\item \textbf{Parallelization.}  First we observe that \eqref{eq:atomicnorm_sparsecoding} is decoupled across the data variables $\by^{(i)}$ for $1 \leq i \leq n$.  As such, the update step \eqref{eq:atomicnorm_sparsecoding} can be performed in parallel.
\item \textbf{Solve \eqref{eq:atomicnorm_sparsecoding} approximately.}  Alternatively, one can attempt to solve \eqref{eq:atomicnorm_sparsecoding} approximately or even very crudely at each iteration.  In our numerical experiments in Sections \ref{sec:numexp}, we solve \eqref{eq:atomicnorm_sparsecoding} via first order methods in which we apply a very modest amount of iterations.  Our numerical experiments, and in particular those on synthetic data in Section \ref{sec:numexp}, suggest that solving \eqref{eq:atomicnorm_sparsecoding} crudely is frequently sufficient to make progress in the overall algorithm.
\item \textbf{Convex relaxations.}  A different approach to solving \eqref{eq:atomicnorm_sparsecoding} approximately is via convex relaxations in which we replace the penalty function $\|\cdot \|_{\grp}$ with an different penalty function that leads to a computationally cheaper program in \eqref{eq:atomicnorm_sparsecoding}.  Concretely, let $\tilde{\grp}$ be a convex outer approximation of the set $\grp$, and let $\|\cdot\|_{\tilde{\grp}}$ be the resulting induced norm.  We replace the penalty function $\|\cdot \|_{\grp}$ with $\|\cdot\|_{\tilde{\grp}}$ for choices of outer approximations whose induced norm is cheaper to evaluate compared to $\|\cdot \|_{\grp}$.

Convex relaxations are used in a wide range of applications as a principled procedure for obtaining cheaper approximations of intractable optimization instances.  In our context, convex relaxations are useful in settings where the set $\mathrm{conv}(\grp)$ is intractable to describe.  We discuss connections between convex relaxations and the group structure of $\grp$ in Section \ref{sec:conc}.
\end{enumerate}

\textbf{Dictionary update.}  In the second step of each iteration, we update the estimates of the dictionary generators $\{ \ba_j \}_{j=1}^{q}$ while keeping the variables $\{Z_j^{(i)} \}_{j=1,i=1}^{q,n}$ fixed.  As the variables $\{ \ba_j \}_{j=1}^{q}$ only appear in the squared error loss objective function \eqref{eq:GIDL}, the update step reduces to the solution of a least squares system:
\begin{equation} \label{eq:modupdate}
\left(\ba_1,\ldots,\ba_q \right) \leftarrow \underset{\ba_1,\ldots,\ba_q}{\arg \min} ~ \sum_{i=1}^{n} \| \by^{(i)} - \sum_{j} Z_j^{(i)}  \ba_j \|_2^2 .
\end{equation}

Note that it is possible to perform the dictionary update via other means, such as by taking a gradient step.  In fact, in regular dictionary learning, updating dictionary estimates based on least squares is frequently referred to as the Method of Optimal Directions (MOD) \cite{EAH:99} -- the proposed update in \eqref{eq:modupdate} based on least squares is precisely the analog of the MOD in our set-up.
\section{Examples} \label{sec:examples}

In this section, we describe instantiations of our framework \eqref{eq:GIDL} for a series of examples.  Our discussion focuses on: (i) specifying the group $\grp$ that expresses the desired invariance, and (ii) providing a \emph{conic programming} description of the set $\mathrm{conv}(\grp)$.  Using widely available software for solving convex programs, our discussion provides the basic tools necessary to implement our framework in a number of settings.  We conclude this section with a description of additional algorithmic simplifications for certain classes of invariances. 

\subsection{Regular Dictionary Learning} \label{sec:example_regular}

Our first example describes how our framework \eqref{eq:GIDL} expresses Regular Dictionary Learning (DL).  The group $\grp_{\mathrm{reg}} = \{ \pm I \}$ is the identity with its negation.  The linear span of $\grp_{\mathrm{reg}}$ are diagonal matrices with constant entries
\begin{equation*}
Z = \left(
\begin{array}{ccc}
z & &\\
& \ddots & \\
& & z
\end{array}
\right),
\end{equation*}
and the atomic norm is $\| Z \|_{\grp_{\mathrm{reg}}} = |z|$.  

Note that the Regular DL problem is more typically cast as follows
\begin{equation} \label{eq:reg_dl}
\begin{aligned}
\underset{ A \in \mathbb{R}^{d \times q} , \bx^{(i)} 
\in \mathbb{R}^{q}}{\arg \min} & ~ \sum_{i=1}^{n} \left( \frac{1}{2} \| \by^{(i)} - A \bx^{(i)} \|_2^2 + \lambda \|\bx^{(i)}\|_{1}\right) \\
\mathrm{s.t.} & \quad A = [\ba_1| \ldots | \ba_q], \quad \| \ba_j \|_2 = 1 \text{ for all } j.
\end{aligned}
\end{equation}
We can recover the above formulation within \eqref{eq:GIDL} by setting $\grp = \grp_{\mathrm{reg}}$, choosing the number of generators to be $q$, and identifying the constant variable in $Z_j^{(i)}$ with $\bx^{(i)}_j$.

\subsection{Integer Shift Invariance / Convolutional Dictionary Learning} \label{sec:example_conv}

Our second example concerns integer shift invariance, and describes how our framework \eqref{eq:GIDL} expresses Convolutional DL.  Assume that the data resides in $\mathbb{R}^{d}$.  Recall that a circulant matrix takes the following form
\begin{equation*}
\left(
\begin{array}{ccccc}
z_0 & z_{d-1} & \ldots & z_2 & z_1 \\
z_1 & z_0 & z_{d-1} & \ldots & z_2 \\
\vdots & z_1 & z_0 & \ddots & \vdots \\
z_{d-2} & \vdots & \ddots & \ddots & z_{d-1} \\
z_{d-1} & z_{d-2} & \ldots & z_1 & z_0
\end{array}
\right)
\end{equation*}
Let $T_r$ be the circulant matrix whose leading column is a vector whose $r$-th entry is equal to one (here, we use the convention that the leading coordinate is the $0$-th coordinate), and whose remaining entries are equal to zero.  Then a shift by $r$ coordinates can be represented by left multiplication by the matrix $T_r = (T_1)^{r}$.  Subsequently, it follows that $\grp_{\mathrm{Int}} = \{ T_r : 0 \leq r \leq d-1 , r \in \mathbb{Z} \}$. 

Linear combinations of $T_r$'s are circulant matrices, and hence the linear span of $\mathrm{conv}(\grp_{\mathrm{Int}})$ are circulant matrices.  The associated atomic norm $\| \cdot \|_{\grp_{\mathrm{Int}}}$ is given by the absolute sum of its entries
\begin{equation} \label{eq:atomicnorm_circulant}
\| Z \|_{\grp_{\mathrm{Int}}} = \sum_{k=0}^{d-1} |z_k| .
\end{equation}

Note that the Convolutional DL problem is more typically cast as follows
\begin{equation*}
\begin{aligned}
\underset{ \ba_j , \bx^{(i)}_{j} \in \mathbb{R}^{d} }{\arg \min} & ~ \sum_{i=1}^{n} \left( \frac{1}{2} \| \by^{(i)} - \sum \ba_j \ast \bx_j^{(i)} \|_2^2 + \lambda \sum \|\bx^{(i)}_j \|_{1}\right) \\
\mathrm{s.t.} & \qquad \| \ba_j \|_2 = 1 \text{ for all } 1 \leq j \leq q.
\end{aligned}
\end{equation*}
We can recover the above formulation within \eqref{eq:GIDL} by setting $\grp = \grp_{\mathrm{Int}}$, choosing the number of generators to be $q$, and identifying the leftmost column of $Z_j^{(i)}$ with $\bx^{(i)}_j$.

\subsection{Continuous Shift Invariance} \label{sec:example_ctsshift}

Our third example concerns a continuous analog of shift invariance, and is motivated by applications in time series analysis.  Suppose that we observe a continuous signal at regular (discrete) intervals, and we wish to learn a dictionary that represents small time segments of the signal.  

One drawback of Convolutional DL is that it does not provide us a mechanism of identifying two observations that are derived by sampling the signal at regular intervals, but spaced apart by a \emph{non integral} shift \cite{SFB:19}.  To this end, Song, Flores and Ba expand on the ideas of Convolutional DL and incorporate a grid-refinement step followed by an interpolation step to learn a continuously shift invariant dictionary \cite{SFB:19}.  In the following, we describe a different approach whereby we express the continuous shift invariance from the outset.  Our approach utilizes trigonometric interpolation and it relies on the existence of semidefinite programming (SDP) representations of certain sets.  In Section \ref{sec:example_convinterp}, we discuss how the framework in \cite{SFB:19} may be viewed as one that interpolates Convolutional DL and our ideas.  

To simplify the exposition, we assume that the data is in $\mathbb{R}^{2d+1}$; i.e., the data dimension is odd.  There is a minor difference in the description of our framework depending on whether the data dimension is odd or even, and this arises because of the way trigonometric interpolation is applied.  

\subsubsection{Description of $\grp_{\mathrm{cts}}$}  Let $\omega = \exp(2i\pi / (2d+1))$, and denote the (normalized) DFT matrix
\begin{equation} \label{eq:DFT_matrix_ctsshift}
F = \frac{1}{\sqrt{d}} 
\left(
\begin{array}{cccc}
\omega^{(-d) \cdot 0} & \omega^{(-d) \cdot 1} & \ldots & \omega^{(-d)  \cdot (2d)} \\
\omega^{(-d+1) \cdot 0} & \omega^{(-d+1) \cdot 1} & \ldots & \omega^{(-d+1) \cdot (2d)} \\
\vdots & \vdots & & \vdots \\
\omega^{(d) \cdot 0} & \omega^{(d) \cdot 1} & \ldots & \omega^{(d) \cdot (2d)} 
\end{array}
\right).
\end{equation}
Let $l(\phi) \in \mathbb{C}^{(2d+1) } $ be the vector whose entries are $ ( \omega^{(-d) \cdot \phi}, \omega^{(-d+1) \cdot \phi}, \ldots, \omega^{(0) \cdot \phi}, \omega^{(1) \cdot \phi} , \ldots, \omega^{(d) \cdot \phi} ) $, and let $L(\phi)$ be the diagonal matrix whose entries are $l(\phi)$.  We say that two vectors $\tilde{\bx} \sim \bx$ if
\begin{equation} \label{eq:mastergrp}
\tilde{\bx} = (F^{\dagger} L(\phi) F) \bx, \qquad \phi \in [0,1).
\end{equation}  
Subsequently, we have $\grp_{\mathrm{cts}} = \{ (F^{\dagger} L(\phi) F) : \phi \in [0,1) \}$.  

\subsubsection{Description of atomic norm}  We proceed to provide a SDP description of the atomic norm induced by $\grp_{\mathrm{cts}} $.  First, note that since the transformation $F$ is unitary, and that the matrix $L$ is diagonal, it suffices to seek descriptions of the set $\mathrm{conv}(\{l (\phi) : \phi \in [0,1) \})$.  
The SDP-based description of the atomic norm relies on useful properties concerning matrices that are positive semidefinite (PSD) Toeplitz.  To this end, we state an intermediate result that summarizes the connection between such matrices and continuous shifts.  Denote $ \bv (\phi) := \left( \exp(i (2\pi) 0 \cdot \phi) , \exp(i (2\pi) 1 \cdot \phi) , \ldots , \exp(i (2\pi) d \cdot \phi)\right)$.  In the following, we use the symbol $+$ ($-$) as a sub-script to represent linear sums of atoms with positive (negative) coefficients.

\begin{proposition}\label{thm:ctsshift_sdp}
Let $\bx \in \R \times \mathbb{C}^{d}$.  Then the minimal values of the following optimization instances are equal:
\begin{equation} \label{eq:ctsshift_form1}
\inf ~ \{  \sum |c_i| ~:~ \bx = \sum_{i \in \mathcal{I}} c_i \bv (\phi_i)  , c_i \in \mathbb{R}, \phi_i \in [0,1] \},
\end{equation}
and
\begin{equation} \label{eq:ctsshift_form2}
\begin{aligned}
\inf ~ \{  z_+ + z_- :
& \bx =  
\left( \begin{array}{c}
z_+ \\ \bz_{+}
\end{array}\right) - 
\left( \begin{array}{c}
z_{-} \\ \bz_{-}
\end{array}\right),  \\
&\left( \begin{array}{cc}
z_+ & \bz_{+}^{\dagger} \\ \bz_{+} & \star 
\end{array} \right),
\left( \begin{array}{cc}
z_- & \bz_{-}^{\dagger} \\ \bz_{-} & \star 
\end{array} \right) \emph{ are PSD Toeplitz}
  \}.
\end{aligned}
\end{equation}
Here, $\mathcal{I}$ is an arbitrary index set, $z_+, z_- \in \mathbb{R}$, $\bz_+,\bz_- \in \mathbb{C}^{d}$, and $\star \in \mathbb{C}^{d \times d}$ denotes an arbitrary matrix.  
\end{proposition}

Proposition \ref{thm:ctsshift_sdp} is based on a similar description for an atomic norm for signals that possess continuous shift invariance \emph{and} phase invariance \cite{TBSR:13}.  The proof of Proposition \ref{thm:ctsshift_sdp} relies on the existence of a Vandermonde decomposition for every PSD Toeplitz matrix, and is likewise based on a similar result in \cite{TBSR:13}.

\begin{proposition}[\cite{CarFej:11,Caratheodory:11,Toeplitz:11} ] \label{thm:vdwfactorization}
Let $X \in \mathbb{C}^{d \times d}$ be a positive semidefinite Toeplitz matrix.  Then $X$ admits a Vandermonde decomposition of the form $X = VDV^{\dagger}$, where $V$ is a $d \times d^{\prime}$ Vandermonde matrix of the form
\begin{equation*}
V = \left(
\begin{array}{cccc}
e^{ i \theta_1 \cdot 0} & e^{i \theta_2 \cdot 0} & \ldots & e^{ i \theta_{d^{\prime}} \cdot 0 } \\ 
e^{ i \theta_1 \cdot 1} & e^{ i \theta_2 \cdot 1} & \ldots & e^{ i \theta_{d^{\prime}} \cdot 1} \\
\vdots & \vdots & & \\
e^{ i \theta_1 \cdot (d-1) } & e^{ i \theta_2 \cdot (d-1) } & \ldots & e^{ i \theta_{d^{\prime}} \cdot (d-1) }
\end{array}
 \right),
\end{equation*}
and $D$ is a diagonal matrix with positive entries. 
\end{proposition}

\begin{proof}[Proof of Proposition \ref{thm:ctsshift_sdp}]
Let $\mathrm{OPT}_1$ be the optimal value to \eqref{eq:ctsshift_form1}, and let $\mathrm{OPT}_2$ be the optimal value to \eqref{eq:ctsshift_form2}.

We begin by showing that $\mathrm{OPT}_1 \geq \mathrm{OPT}_2$.  Suppose $\bx = \sum_{j \in \mathcal{J}} c_j \bv(\phi_j) - \sum_{k \in \mathcal{K}} d_k \bv(\phi_k)$ where $c_j,d_k > 0$ for all $j \in \mathcal{J}, k \in \mathcal{K}$.  Construct the matrices $Z_+ = \sum_{j \in \mathcal{J}} c_j \bv(\phi_j) \bv(\phi_j)^{\dagger}$ and $Z_- = \sum_{k \in \mathcal{K}} d_k \bv(\phi_k) \bv(\phi_k)^{\dagger}$.  It is easy to see that $Z_+$ and $Z_-$ are PSD Toeplitz, and that the first column of $Z_+ - Z_-$ is precisely $\bx$; i.e., $Z$ is a feasible matrix in \eqref{eq:ctsshift_form2}.  By taking the infimum over all possible decompositions of the form $\bx = \sum_{j \in \mathcal{J}} c_j \bv(\phi_j) - \sum_{k \in \mathcal{K}} d_k \bv(\phi_j)$, it follows that $\mathrm{OPT}_1 \geq \mathrm{OPT}_2$.

Next, we show that $\mathrm{OPT}_2 \geq \mathrm{OPT}_1$.  Let $Z_+$ and $Z_-$ be the respective PSD Toeplitz matrices in \eqref{eq:ctsshift_form2}.  By Proposition \ref{thm:vdwfactorization}, $Z_+$ and $Z_-$ admit a Vandermonde decompositions of the form $Z_+ = \sum_{j \in \mathcal{J}} c_j \bv(\phi_j) \bv(\phi_j)^{\dagger}$ and $Z_- = \sum_{k \in \mathcal{K}} c_k \bv(\phi_k) \bv(\phi_k)^{\dagger}$.  Note that we have $\bx = \sum_{j \in \mathcal{J}} c_j \bv(\phi_j) - \sum_{k \in \mathcal{K}} d_k \bv(\phi_j)$.  By taking the infimum over all feasible solutions to \eqref{eq:ctsshift_form2}, and by noting that $z = \sum c_j - \sum d_k$, it follows that $\mathrm{OPT}_2 \geq \mathrm{OPT}_1$.
\end{proof}

Let $\mathbb{W}^{2d+1} = F(\mathbb{R}^{2d+1})$ be the image of the real vector space under $F$.  One can verify $ \mathbb{W}^{2d+1} = \{ \, (f_{-d},\ldots,f_{-1},f_0,f_{1},\ldots,f_d)  \,:\, f_i = (f_{-i})^{\ast}, \, -d \leq i \leq d \, \} \subset \mathbb{C}^{2d+1}$.  Then the linear span of $\mathrm{conv}(\grp_{\mathrm{cts}})$, after a change of basis by the unitary transformation $F$, are diagonal matrices whose diagonal are in $\mathbb{W}^{2d+1}$.  The following result summarizes the description of the atomic norm.
\begin{proposition} \label{thm:atomicnorm_fullsdp}
Let $\bx \in \mathbb{R}^{2d+1}$.  Then
\begin{equation*}
\begin{aligned}
\| \bx \|_{\grp_{\mathrm{cts}}} ~&=~ \underset{z_{+} , z_{-} \in \R, ~ \bz_{+},\bz_{-} \in \mathbb{C}^{d}}{\inf}~~  z_{+} + z_{-} \\ 
\mathrm{s.t.} ~~&  F \bx = \left( \begin{array}{c}  \bz_{+}^{\ast} \\ z_+ \\ \bz_{+} \end{array}\right)
- \left( \begin{array}{c} \bz_{-}^{\ast} \\ z_- \\ \bz_{-} \end{array}\right) \\
& \left( 
\begin{array}{cc}
z_{+} & \bz_{+}^{\dagger} \\
\bz_{+} & \star
\end{array}
\right),
 \left( 
\begin{array}{cc}
z_{-} & \bz_{-}^{\dagger} \\
\bz_{-} & \star
\end{array}
\right) \emph{are PSD Toeplitz} .
\end{aligned}
\end{equation*}
\end{proposition}

\subsubsection{Implementation details} \label{sec:ctsshift_implement}

In the following, we explain how Proposition \ref{thm:atomicnorm_fullsdp} is applied within our dictionary learning framework.  Let $\by \in \R^{2d+1}$ be the data vector and let $\{\ba_i\}_{i=1}^{q} \in \R^{2d+1}$ be generators for the dictionary.  Let $\tilde{\by}:= F \by$ and $\tilde{\ba}_i := F \ba_i$ be the data and the generator expressed in the transformed basis, and let $\odot$ denote pointwise multiplication.  Then the loss can be re-written as follows: $\| \by - \sum_{i} \left( F^{\dagger} L(\phi) F \right) \ba_i \|_2^2 = \| (F \by) - \sum_{i} L(\phi) (F \ba_i) \|_2^2 =  \| \tilde{\by} - \sum_{i} l(\phi) \odot \tilde{\ba}_i \|_2^2$.

Following Proposition \ref{thm:atomicnorm_fullsdp}, the step in which we represent $\by$ as succinctly with respect to $\mathcal{D}$ as in \eqref{eq:atomicnorm_sparsecoding} is given by:
\begin{equation} \label{eq:ctsshiftsparsecoding}
\begin{aligned}
& \underset{\substack{ z_{i,+},z_{i,-} \in \R,\\ \bz_{i,+},\bz_{i,-} \in \mathbb{C}^{d}}}{\mathrm{argmin}} \frac{1}{2} \| \tilde{\by} - \sum_{i=1}^{q} \tilde{\ba}_i \odot (
\left(\begin{array}{c}
\bz_{i,+}^{\ast} \\ z_{i,+} \\ \bz_{i,+} 
\end{array}\right)
-
\left(\begin{array}{c}
\bz_{i,-}^{\ast} \\ z_{i,-} \\ \bz_{i,-}
\end{array}\right)
) \|_2^2  + \lambda \sum_{i=1}^{q} (z_{i,+} + z_{i,-}) \\
&  \mathrm{s.t.}
\left( \begin{array}{cc}
z_{i,+} & \bz_{i,+}^{\dagger} \\
\bz_{i,+} & \star
\end{array}\right),
\left( \begin{array}{cc}
z_{i,-} & \bz_{i,-}^{\dagger} \\
\bz_{i,-} & \star
\end{array}\right) \text{are PSD Toeplitz}
\end{aligned}.
\end{equation}

\subsection{Integer Shift Invariance with Interpolation} \label{sec:example_convinterp}

In our fourth example, we expand on the techniques in Section \ref{sec:example_ctsshift} to describe an approach for learning continuously shift invariant dictionaries using interpolation.  Our discussion is conceptually identical to the framework proposed by Song, Flores, and Ba \cite{SFB:19}.

More concretely, suppose we restrict the shift parameter $\phi$ in \eqref{eq:mastergrp} to be integer multiples of $1/((2d+1)K)$ instead of all values in the interval $[0,1)$, as in Section \ref{sec:example_ctsshift}.  Here, $K$ is an integer parameter, and it denotes the number of subdivisions within a single integer shift.  Then, integer shift invariance can be viewed as one extreme where $K=1$, and continuous shift invariance can be viewed as the other extreme where $K = \infty$.  We then have $\grp_{\mathrm{interp}} = \{ (F^{\dagger} L(\phi) F) : \phi = k/((2d+1)K), 0 \leq k \leq ((2d+1)K)-1\}$.  The resulting atomic norm is polyhedral, and one can show that the resulting DL problem is equivalent to the following
\begin{equation*} 
\begin{aligned}
\underset{ A , \bx^{(i)} }{\arg \min} & ~ \sum_{i=1}^{n} \left( \frac{1}{2} \| \by^{(i)} - A \bx^{(i)} \|_2^2 + \lambda \|\bx^{(i)}\|_{1}\right) .
\end{aligned}
\end{equation*}
Here, the linear map $A$ comprises all interpolated shifts of the dictionary generators, and is subsequently of size $d$ by $(q \times (2d+1) \times K)$.

The relationship between Convolutional DL, the interpolated variant we describe here, and continuously shift invariant DL (Section \ref{sec:example_ctsshift}) suggests the following rule of thumb:  For small $K$, it is preferable to apply the interpolated variant as the sub-routine \eqref{eq:atomicnorm_sparsecoding} entails solving a moderately larger LP instead of a SDP.  For large $K$ however, it is preferable to learn a continuously shift invariant dictionary because the size of the convex program associated to \eqref{eq:atomicnorm_sparsecoding} does not grow with $K$ (i.e., it does not grow with finer discretizations).  

Finally, we note that the framework proposed by Song, Flores, and Ba \cite{SFB:19} is far more general; in particular, it permits more general interpolation methods as well as signal reconstruction methods.  Our current discussion is simply an adaptation of their ideas, and we do so in order to seamlessly describe the relationship between our work and theirs.

\subsection{Invariance to Orthogonal Transformations} \label{sec:example_orthogonal}

In our fifth example, the data are \emph{matrices} $\{Y^{(i)}\}_{i=1}^{n} \subset \mathbb{R}^{d \times r}$, and our goal is to learn a dictionary that is invariant under \emph{orthogonal transformations}.  The generators $\{A_j\}_{j=1}^{q}$ are matrices of dimensions $d \times r$, and the group $\mathfrak{G}_{\mathrm{orth}} = \{ Q : Q \in SO(d) \}$ is the collection of orthogonal matrices acting by left multiplication; that is, $Q A$ is a dictionary element for any orthogonal $Q$ if $A$ is.  To learn a dictionary that is invariant to orthogonal shifts, we minimize the following
\begin{equation} \label{eq:sync_gidl}
\underset{ Z^{(i)}_j, A_j }{\arg \min}  ~ \sum_{i=1}^{n} \frac{1}{2} \| Y^{(i)} - \sum_j Z^{(i)}_j A_j \|_2^2 + \lambda \cdot \sum \|Z^{(i)}_j\|_{\mathfrak{G}_{\mathrm{orth}}} .
\end{equation}
The atomic norm $ \|\cdot\|_{\mathfrak{G}_{\mathrm{orth}}}$ induced by the collection of orthogonal matrices is the \emph{spectral} norm.  To compute the proximal map of a matrix $Z$ with respect to the spectral norm and parameter $\lambda$, we perform the following operations to its singular value decomposition (SVD): (i) replace the largest singular value $s_{\max}$ by $\tilde{s} := \max \{ 0, s_{\max} - \lambda \}$, and (ii) replace all other singular values by $\tilde{s}$ if it exceeds $\tilde{s}$.

\subsection{Algorithmic Simplifications}

We describe a first order method for minimizing \eqref{eq:atomicnorm_sparsecoding} for the specific instances where the invariance is $\grp_{\mathrm{cts}}$.  We recognize that the minimization instance is a SDP, and hence can be solved using standard software \cite{Ren:01,NesNem:94}.  Nevertheless, we believe that there is value in stating our algorithm because it exploits certain structural properties of Hermitian and Toeplitz matrices, and it relies on very basic linear algebraic computations.  In the following, we let $\mathbb{H}^{d}$ denote the set of $d \times d$ complex Hermitian matrices, and we let $\mathbb{H}^{d}_{+}$ denote the set of $d \times d$ complex Hermitian PSD matrices.  In addition, we let $\mathbb{T}^{d}$ denote the set of $d\times d$ complex Hermitian Toeplitz matrices.  

We proceed by discussing the simplifications for $\grp_{\mathrm{cts}}$.  First, we express \eqref{eq:ctsshiftsparsecoding} in following form for some $f$
\begin{equation} \label{eq:generic_fomshape}
\min_{Z \in \mathbb{H}^{d}} \quad f(Z) \qquad \text{s.t.} \qquad Z \in \mathbb{H}^{d}_{+} \cap \mathbb{T}^{d}.
\end{equation}
Our algorithm is based on projected gradient descent in which we alternate between taking steps in the negative direction of the gradient and a projection onto the subset $\mathbb{H}^{d}_{+} \cap \mathbb{T}^{d}$.  We perform the latter step using an alternating projections-based method in which we alternate between applying projection operations onto the sets $\mathbb{H}^{d}_{+}$ and $\mathbb{T}^{d}$ \cite{BD:86,SufHay:93}.

Projected gradient descent-based methods are most effective whenever the projection step is simple to compute.  While such methods are most typically applied whenever the projection step is expressible via a closed-form expression, we are not aware if projections onto $\mathbb{H}^{d}_{+} \cap \mathbb{T}^{d}$ can be expressed as such.  Instead, we settle on a slightly more expensive operation based on alternating projections because projections onto the sets $\mathbb{H}^{d}_{+}$ and $\mathbb{T}^{d}$ only require simple primitives.  More specifically, let $X$ be a Hermitian matrix with eigendecomposition $U D U^{\dagger}$.  The projection of $X$ onto $\mathbb{H}^{d}_{+}$ is the matrix $U D_{+} U^{\dagger}$, where $D_{+}$ is the diagonal matrix obtained as the pointwise maximum between $D$ and the zero matrix \cite{Higham:88}.  The projection of $X$ onto $\mathbb{T}$ is a \emph{linear} operation as $\mathbb{T}$ is a \emph{subspace}.  The procedure for computing projections onto $\mathbb{H}^{d}_{+} \cap \mathbb{T}^{d}$ is based on a more general procedure for computing projections onto the intersection of two convex sets using projections of each of these convex sets as basic primitives.  The more general algorithm was first proposed by Boyle and Dykstra \cite{BD:86}, and subsequently adapted to our set-up by Suffridge and Hayden \cite{SufHay:93}.  We summarize the procedure for computing projections onto $\mathbb{H}^{d}_{+} \cap \mathbb{T}^{d}$ in Algorithm \ref{alg:proj_PSDTOEP}, and we summarize the full procedure for minimizing \eqref{eq:ctsshiftsparsecoding} in Algorithm \ref{alg:compute_ctsSDP}.

\begin{algorithm}[t]
    \caption{Algorithm for computing projection of $X$ onto $\mathbb{H}^{d}_{+} \cap \mathbb{T}^{d}$}
    \label{alg:proj_PSDTOEP}
	\textbf{Input}: A complex Hermitian $d\times d$ matrix $X$.  Initialize $d \times d$ Hermitian matrices $Y,Q = 0$. \\
	\textbf{Algorithm}: Repeat until success iterates differ by at most $\epsilon$. \\
	\textbf{1.} Project $X+P$ onto $\mathbb{H}^{d}_{+}$: \\
	\textbf{-- a.} Compute eigendecomposition $U D U^{\dagger} \leftarrow X+P$. \\
	\textbf{-- b.} Compute pointwise maximum $D \leftarrow \max \{ D,0 \}$. \\
	\textbf{-- c.} Update $Y \leftarrow U D U^{\dagger}$. \\
	\textbf{2.} Update $P \leftarrow X+P-Y$ \\
	\textbf{3.} Project $Y+Q$ onto $\mathbb{T}$: \\
    \textbf{-- a.} Set $X$ to be the Toeplitz matrix with $ X_{1,k} \leftarrow \frac{1}{d-k+1}\sum_{i=1}^{d-k+1} (Y+Q)_{i,k+i-1} + \frac{1}{d-k+1}\sum_{i=1}^{d-k+1} (Y+Q)_{k+i-1,i}^*$, $ 1 \leq k \leq d$ , and
$X_{k,1} \leftarrow \frac{1}{d-k+1}\sum_{i=1}^{d-k+1} (Y+Q)_{i,k+i-1}^* + \frac{1}{d-k+1} \sum_{i=1}^{d-k+1} (Y+Q)_{k+i-1,i}$ , $1 \leq k \leq d$.\\
    \textbf{4.} Update $Q \leftarrow Y+Q-X$ \\
	\textbf{Output}: $X$. 
\end{algorithm}

\begin{algorithm}[t]
    \caption{Algorithm for minimizing \eqref{eq:ctsshiftsparsecoding}}
    \label{alg:compute_ctsSDP}
    \textbf{Algorithm}: Repeat until success iterates differ by at most $\epsilon$ \\
	\textbf{Input}: Step size parameter $\eta > 0$.  Initial matrix estimate $X = 0$. \\
    \textbf{1.} Take negative gradient descent step: $X \leftarrow X - \eta \nabla f(X)$ \\
	\textbf{2.} Compute projection of $X$ onto $\mathbb{H}^{d}_{+} \cap \mathbb{T}^{d}$ using Algorithm \ref{alg:proj_PSDTOEP} \\
	\textbf{Output}: $X$. 
\end{algorithm}

\textbf{Computational complexity.}  The optimization instance in \eqref{eq:ctsshiftsparsecoding} is a SDP of size $O(qd) \times O(qd)$.  Using the procedure in Algorithm \ref{alg:compute_ctsSDP}, the most expensive sub-routine is to project onto $\mathbb{H}^{d}_{+}$, which has complexity $O(d^3)$ using regular matrix eigenvalue decompositions.  Consequently, the overall complexity cost per inner loop iteration is $O(q d^3)$, and the overall complexity of the proposed first order method for solving \eqref{eq:ctsshiftsparsecoding} has complexity $O(q d^3 m_{1} m_{2})$, where $m_{1}$ and $m_{2}$ are the number of inner and outer loops in Algorithm \ref{alg:proj_PSDTOEP} respectively.
\section{Numerical Experiments} \label{sec:numexp}

In this section, we apply our framework over a series of numerical experiments using synthetically generated data as well as real data.  We discuss the utility as well as the limitations of our framework.

\textbf{Dictionary distances.}  We use the following measure to describe the distance between two dictionaries.  Let $\mathcal{D}$ be a dictionary whose generators are $\mathcal{A} = \{ \ba_{1},\ldots, \ba_{q} \}$, and let $\mathcal{D}^{\prime}$ be another arbitrary dictionary.  We define distance between $\mathcal{D}$ and $\mathcal{D}^{\prime}$ in terms of the mean squared error between every generator in $\mathcal{A}$ from an element in $\mathcal{D}^{\prime}$:
\begin{equation} \label{eq:dictdistance}
\mathrm{dist}(\mathcal{D},\mathcal{D}^{\prime}) = \frac{1}{q} \sum_{i=1}^{q} \inf_{\bd \in \mathcal{D}^{\prime}} \| \ba_i - \bd \|_2^2. 
\end{equation}

\textbf{Implementation details.}  We briefly describe the implementation details used in the numerical experiments.  We solve the step corresponding to \eqref{eq:atomicnorm_sparsecoding} via a first-order method (we provide details shortly).  To ensure that our comparison of different dictionary learning frameworks is fair, the first-order methods we deploy are all comparable, and the number of iterations is always set equal to $5$.  Our rule of thumb for choosing the regularization parameter $\lambda$ is to select it as large as possible while ensuring that the matrices $Z_j^{(i)}$'s are not degenerate.  The intuition is to maximize the impact of the structure inducing penalty terms.  In particular, this rule of thumb appears to be most useful in real data in that smaller choices of $\lambda$ (than those specified by our rule of thumb) tend to learn less meaningful dictionary elements.

Our implementation of Regular DL follows the description in Section \ref{sec:example_regular}, and specifically the formulation in \eqref{eq:reg_dl}.  We solve the sub-routine corresponding to \eqref{eq:atomicnorm_sparsecoding} via a first-order method in which we alternate between a gradient step (with a line search) followed by a proximal step with respect to the induced atomic norm (i.e. the L1-norm).  

Our implementation of integer shift invariant / Convolutional DL follows the description in Section \ref{sec:example_conv}.  We solve the sub-routine corresponding to \eqref{eq:atomicnorm_sparsecoding} using the same procedure as in Regular DL.

Our implementation of continuous shift invariant dictionary learning follows the description in Section \ref{sec:example_ctsshift}.  We solve the sub-routine corresponding to \eqref{eq:atomicnorm_sparsecoding} using projected gradient descent (with a line search) as described in Section \ref{sec:ctsshift_implement}.  We apply $5$ outer iterations of the projected gradient descent, and in each iteration we apply a single iteration in the inner loop (where we project a collection of matrices to be PSD Toeplitz).
\subsection{Incorporating Invariances} \label{sec:shiftinvariantdiscrete}

Our first example is on synthetic data and it expands on the experiment described in the Introduction.

\textbf{Signal model.}  We draw $q=3$ unit Euclidean-norm $\{ \ba^{\star}_j : 1 \leq j \leq q \} \subset \R^{d}$ with $d=30$ from the uniform measure as our generators.  We define the dictionary $\mathcal{D}^{\star} = \{ T^{r} \cdot \ba^{\star}_j : \bd_j \in \R^d, 1 \leq j \leq q, 0 \leq r \leq d-1 \}$ to include all possible integer shifts.  We generate $n=10000$ data-points according to the following model:
\begin{equation*}
	\by^{(i)} = \sum_{j=1}^{s} c^{(i)}_j \bd^{(i)}_j, \qquad \bd_j \in \mathcal{D}^{\star}.
\end{equation*}
Here, the coefficients $c_j \sim \mathcal{N}(0,1)$ are i.i.d. normal random variables, and the dictionary elements $\bd_j$ are chosen from the dictionary $\mathcal{D}^{\star}$ uniformly at random (u.a.r.).  We choose the sparsity parameter $s=5$.

\textbf{Integer shift invariance.}  First, we apply Convolutional DL as described in Section \ref{sec:example_conv} on a subset of only $1000$ data-points.  We apply $50$ iterations and we supply the choice of regularization parameter $\lambda =0.4$ as well as the correct number of generators $q=3$.  
We plot the error between each iterate and the true dictionary $\mathcal{D}^{\star}$ in Figure \ref{fig:ComparisonRecovery} (see the dashed lines in both plots).  We repeat this experimental set-up over $10$ different random initializations.  In all instances, we observe that our algorithm recovers the underlying dictionary with an error of approximately $0.05$.

\textbf{Comparison with Regular Dictionary Learning.}  Second, we compare the results with Regular DL as described in Section \ref{sec:example_regular}.  We specify the number of generators to be equal to $3 \times 30$, which is actual number of dictionary elements, and we perform $100$ iterations.  In the left sub-plot of Figure \ref{fig:ComparisonRecovery} we supply the same $1000$ data-points, and we observe that in the algorithm recovers a dictionary with error $\approx 0.2$ from the underlying dictionary.  In the right sub-plot of Figure \ref{fig:ComparisonRecovery} we supply all $10000$ data-points, and we observe that the error improves to $0.1$ in the latter set-up, which is still poorer than the results obtained using Convolutional DL using \emph{a tenth} of the dataset.

These results emphasize the importance of incorporating the appropriate invariant structure, particularly in settings where data is limited.  A plausible explanation for this phenomenon is that these invariances help reduce the degrees of freedom in the estimation problem significantly.

\begin{figure}
	\centering
	\includegraphics[width=0.45\textwidth]{images/Versus_GIDLCDL_ErrOracle}
	\includegraphics[width=0.45\textwidth]{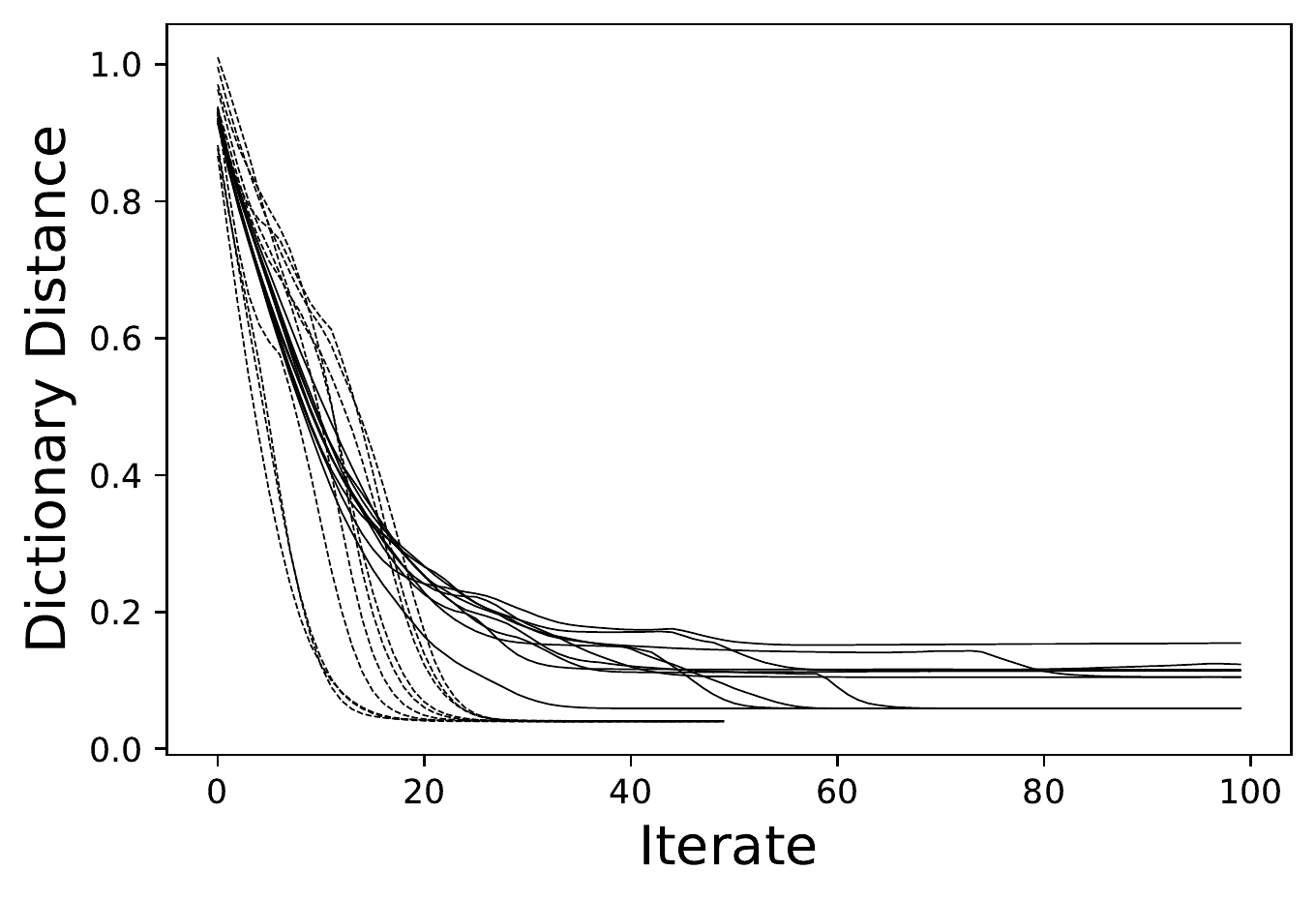}
	\caption{Comparison of learning a shift invariant dictionary using an algorithm that incorporates shift invariance as a prior (our framework -- dashed lines) with an algorithm that does not (Regular DL -- solid lines).  In the left sub-plot we compare our framework with regular dictionary learning on the same dataset comprising $1000$ data-points; in the right sub-plot we compare our framework using $1000$ data-points with Regular DL applied over $10000$ data-points.}
	\label{fig:ComparisonRecovery}
\end{figure}
\subsection{Learning Shift Invariant Dictionaries for ECG Data} \label{sec:ecg_ctsshift}

We apply our framework in Section \ref{sec:example_ctsshift} to learn a shift invariant dictionary for an ECG time series.  The dataset $\{ \by^{(j)} \}_{j=1}^{n} \subset \R^{201}$ contains $n=1000$ time series of length $201$ segmented from a longer time series, which is an ECG signal sampled at 360Hz obtained from the MIT-BIH Arrhythmia Database (signal 100.dat) \cite{mitbih-physionet:01,physionet:00}.  We subtract an offset from each time series so that it is zero mean, and we scale the resulting vector to be unit-norm.  Figure \ref{fig:SampleECGa} shows a subset of our dataset $\{ \by^{(j)} \}_{j=1}^{n}$, and Figure \ref{fig:SampleECGb} shows a segment of the longer time series.

We apply our method to learn a dictionary that is continuously shift invariant with $q \in \{ 1,2,3,4 \}$ generators.  We apply $20$ iterations of our method.  For $q \in \{1,2\}$ we pick $\lambda = 0.2$, and for $q \in \{3,4\}$ we pick $\lambda = 0.1$ as our choices of regularization parameter.   In Figure \ref{fig:CtsShift_Range}, we show the realizations of continuous shifts applied to the generators learned from the instance where $q=4$.

We show the generators obtained from our method in Figure \ref{fig:CtsShift_ECG}.  We observe that the generators resemble spikes, which is consistent with what we expect from ECG signals.  Interestingly, we also note that the waveforms appear to be largely consistent across different choices of $q$. 

\textbf{Comparison with other methods.} We compare the learned generators with those obtained using other methods.  First, we apply Convolutional DL (Section \ref{sec:example_conv}) to learn a single generator (Figure \ref{fig:ECG_othermethods}, top left).  In this instance, we apply $20$ iterations, and we pick $\lambda = 0.1$ as our choice of regularization parameter.  Second, we apply an interpolated variant of Convolutional DL (Section \ref{sec:example_convinterp}) to also learn a single generator (Figure \ref{fig:ECG_othermethods}, top right).  We add one interpolation point between every integer shift so that the number of dictionary elements is double that of Convolutional DL.  In this instance, we apply $20$ iterations, and we pick $\lambda = 0.1$ as our choice of regularization parameter.  Third, we apply the same interpolated variant of Convolutional DL, but with \emph{three} additional interpolation points per shift to learn a single atom (Figure \ref{fig:ECG_othermethods}, bottom left).  Subsequently, the number of dictionary elements is quadrupled compared to the first instance with Convolutional DL.  We apply $20$ iterations, and we pick $\lambda = 0.1$ as our choice of regularization parameter.  Fourth, we apply Regular DL (Section \ref{sec:example_regular}) to learn $q=201$ dictionary atoms (Figure \ref{fig:ECG_othermethods}, bottom right).  In this instance, we apply $100$ iterations, and we pick $\lambda = 0.02$ as our choice of regularization parameter.  

For Convolutional DL and its interpolated variants, we observe that the learned generators also resemble spikes, which is consistent with the waveforms learned using a continuously shift invariant dictionary.  
For Regular DL however, we note the presence of numerous dictionary elements that do not resemble spikes.  The probable explanation for this is that Regular DL has substantially more degrees of freedom compared to the other methods we applied, and as a result picks up a substantially higher number of waveform patterns.  This suggests that, if one is specifically interested in learning essential features of a dataset, then incorporating some form of structural invariance is essential.  

We remark that while all methods incorporating shift invariance learn similar looking templates, minor differences exist between the learned templates.  For instance, the generator learned using integer shift invariance has a less pronounced peak (absolute value $\approx 0.5$) compared to the generator learned using continuously shift invariance (absolute value $> 0.8$).  In particular, the generators learned using interpolated integer shifts lie between both extremes (absolute value $\approx 0.7$ using twice as many dictionary elements, and $\approx 0.8$ using four times as many dictionary elements).  It is not entirely clear if these differences are artifacts of random initializations, or reflective of genuine differences in these methods.  One possible explanation we put forward is based on the intuition that signals with high frequency components (as is the case for spikes) can appear quite different between consecutive integer shifts.  As such, methods that do not account for a continuum of shifts needs to suppress the high frequency components so that the learned generator and its integer shifted copies ``cover'' the data well, while methods that do permit continuous shifts are not constrained in a similar way.  Conversely, if the signal is sufficiently smooth, we expect the dictionaries learned using integer shift invariant (Convolutional) DL and continuously shift invariant DL to be qualitatively identical.

\textbf{Comparison in compute time.}  In Figure \ref{fig:Iteration_time}, we compare the per iteration time across all methods.  We specifically record the time taken to solve \eqref{eq:atomicnorm_softthresholding} across all instances, and exclude the time taken for the dictionary update step.  As a note, the experiments were conducted in {\footnotesize PYTHON} on a machine fitted with an Intel Core i7-7600U running at 2.80GHz.  Our results suggest that the interpolated version of Convolutional DL is computationally more expensive than the vanilla Convolutional DL, but cheaper than learning continuously shift invariant dictionaries.  As such, the interpolated version of Convolutional DL as well as the more general framework in \cite{SFB:19} may be preferred to learning continuously shift invariant dictionaries if the number of interpolants required is modest.

\begin{figure}
	\centering
	\includegraphics[width=0.60\textwidth]{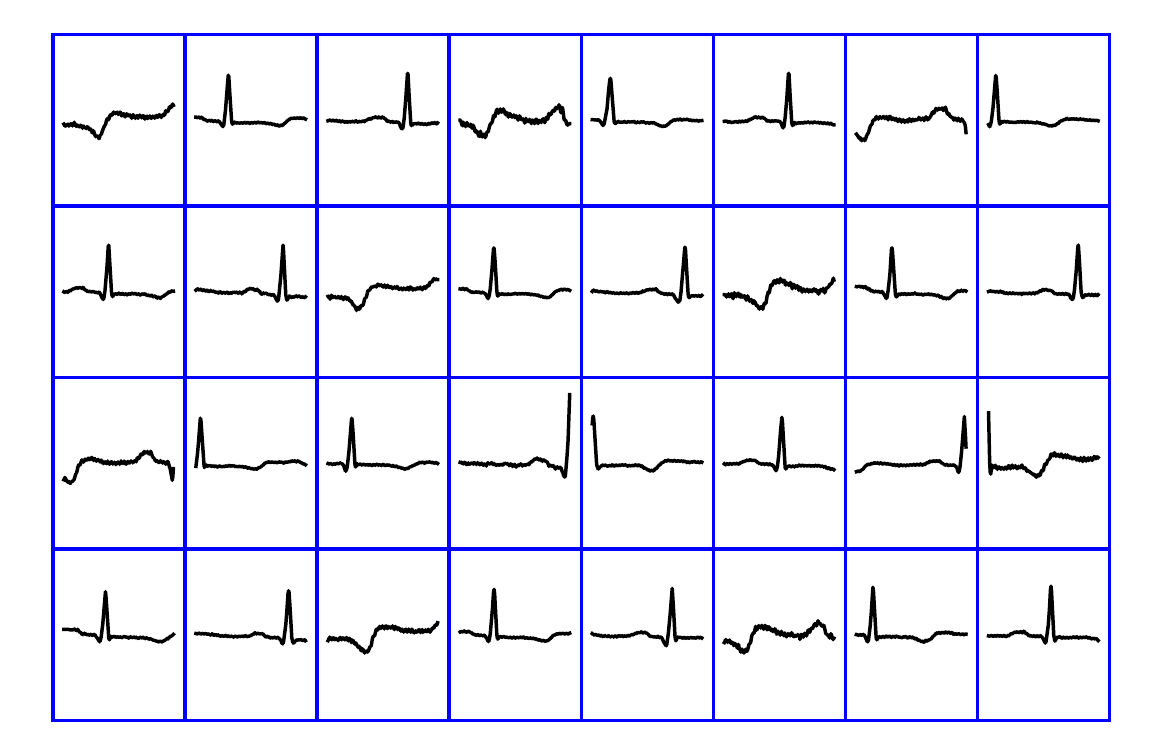}
	\caption{Samples of our dataset.  Each datapoint is a time series of length $101$ segmented from a longer time series.}
	\label{fig:SampleECGa}
\end{figure}

\begin{figure}
	\centering
	\includegraphics[width=0.6\textwidth]{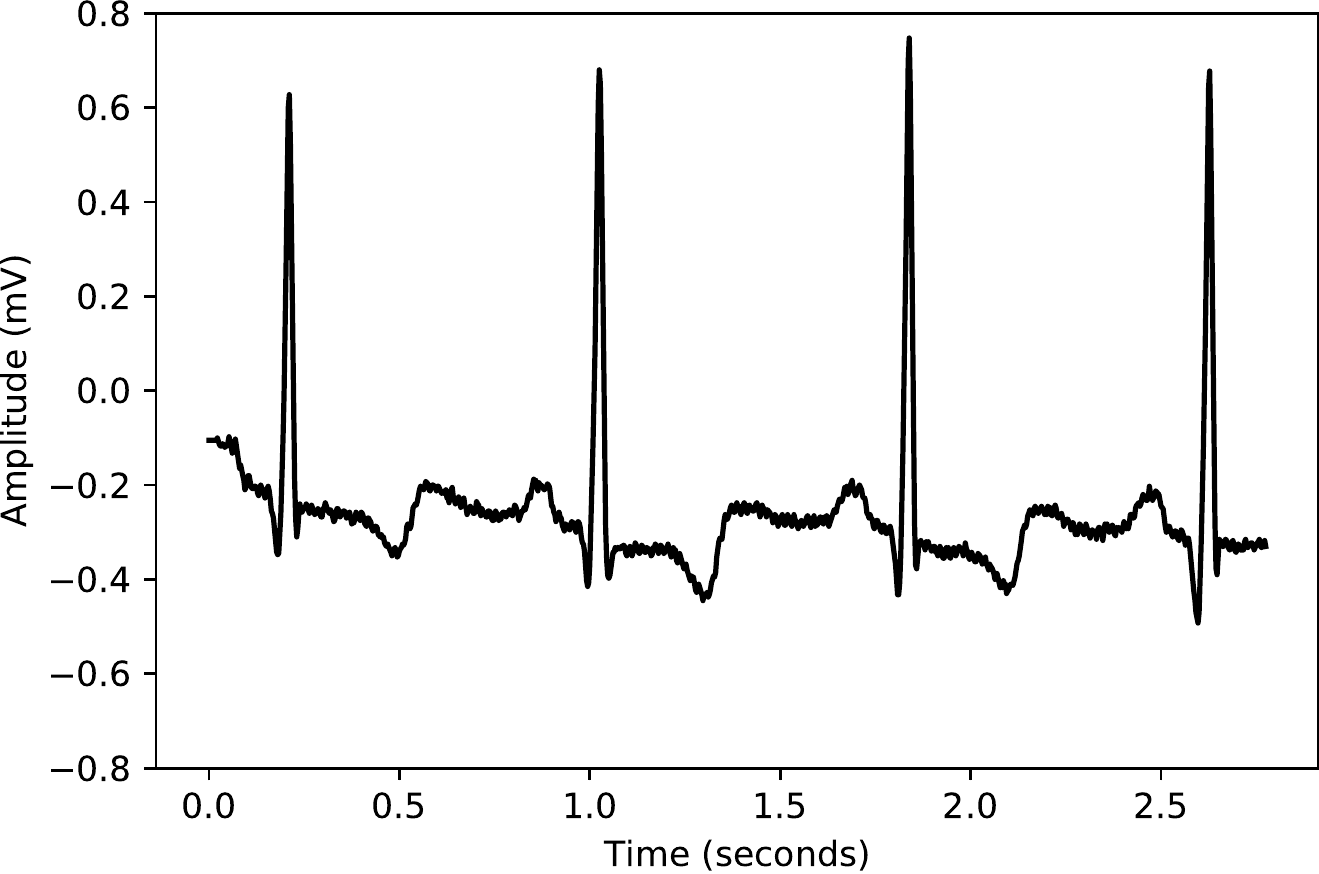}
	\caption{A segment of the longer ECG time series.}
	\label{fig:SampleECGb}
\end{figure}

\begin{figure}
	\centering
	\includegraphics[width=0.45\textwidth]{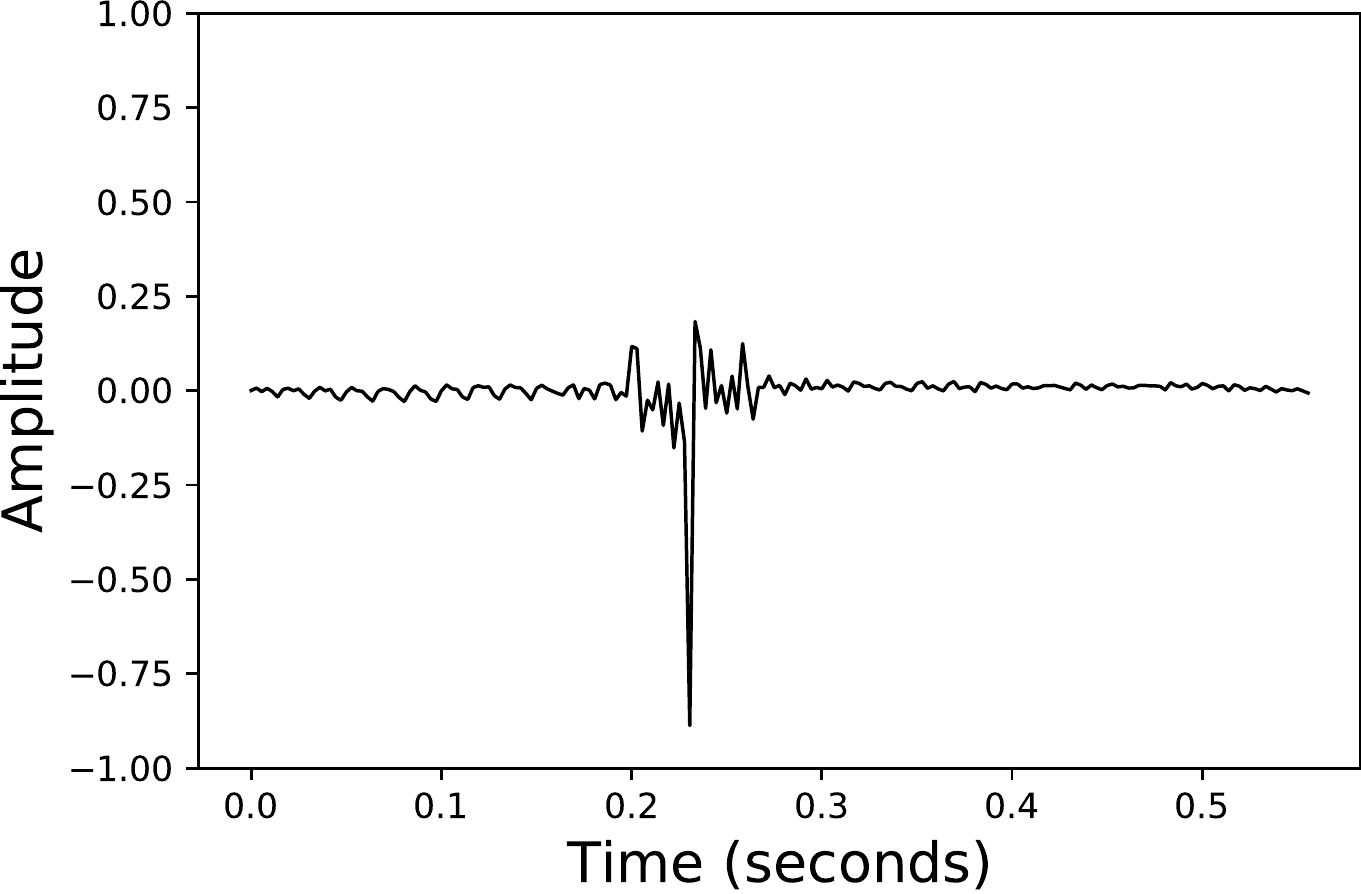}
	\includegraphics[width=0.45\textwidth]{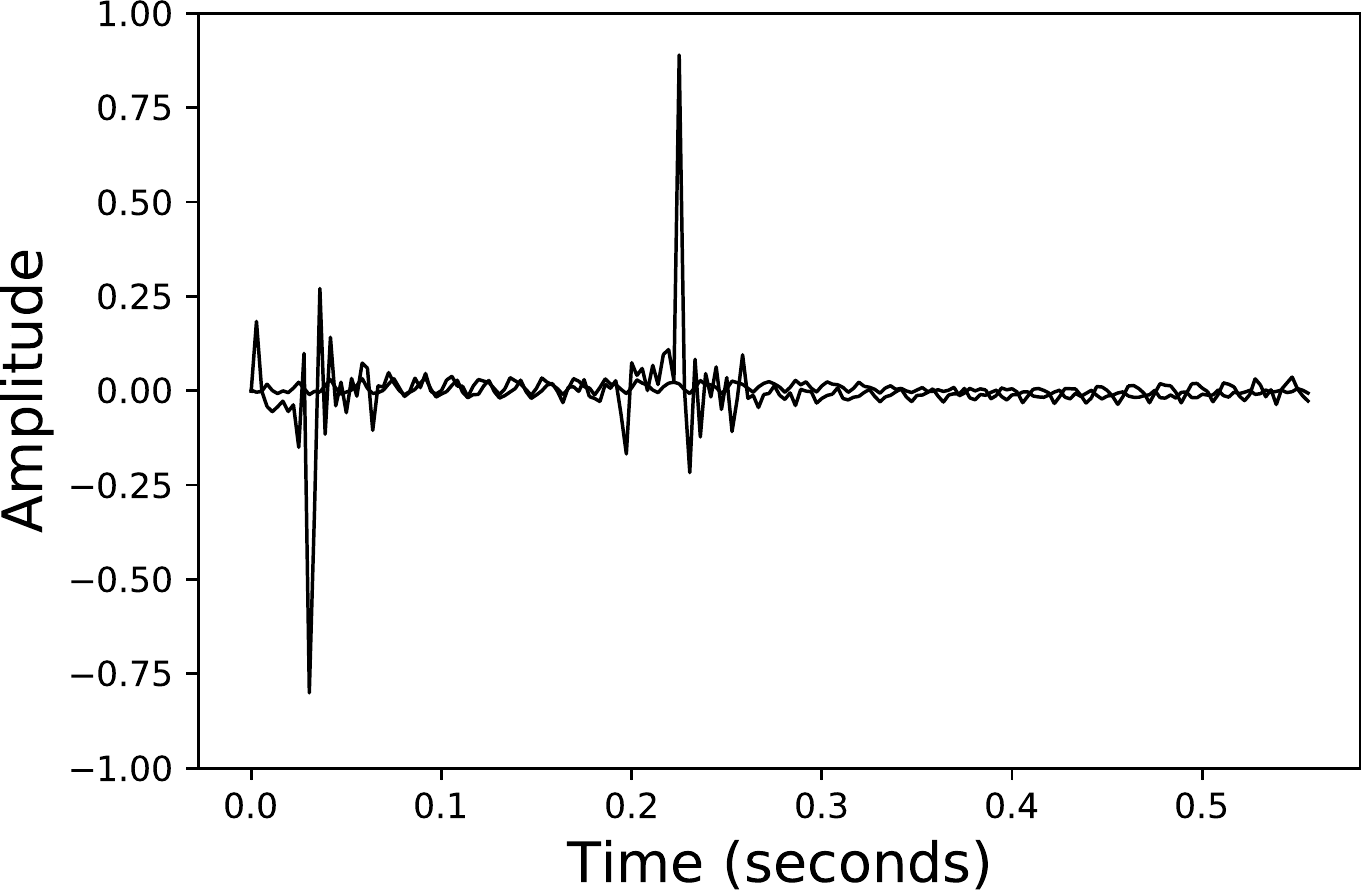}
	\includegraphics[width=0.45\textwidth]{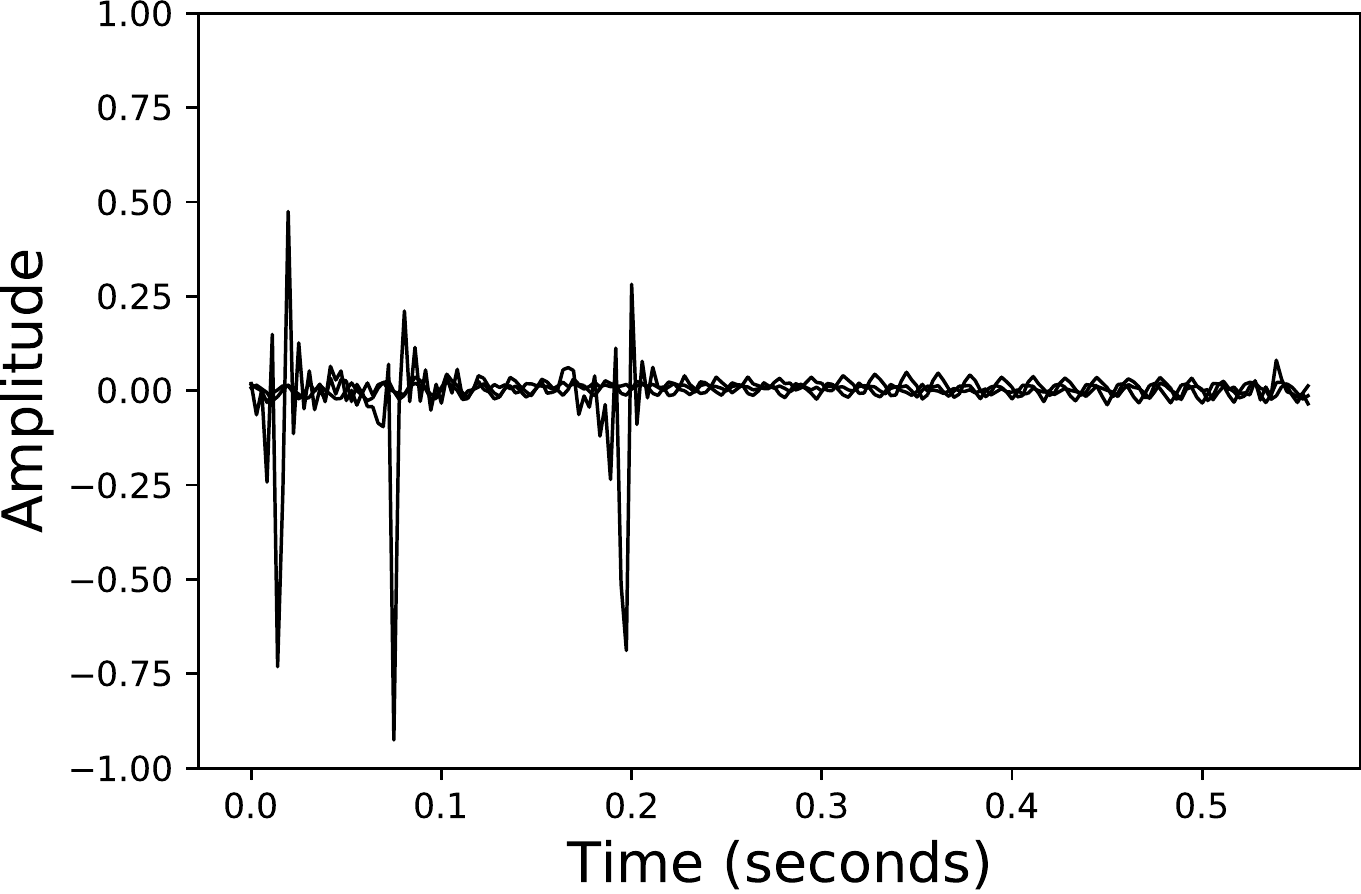}
	\includegraphics[width=0.45\textwidth]{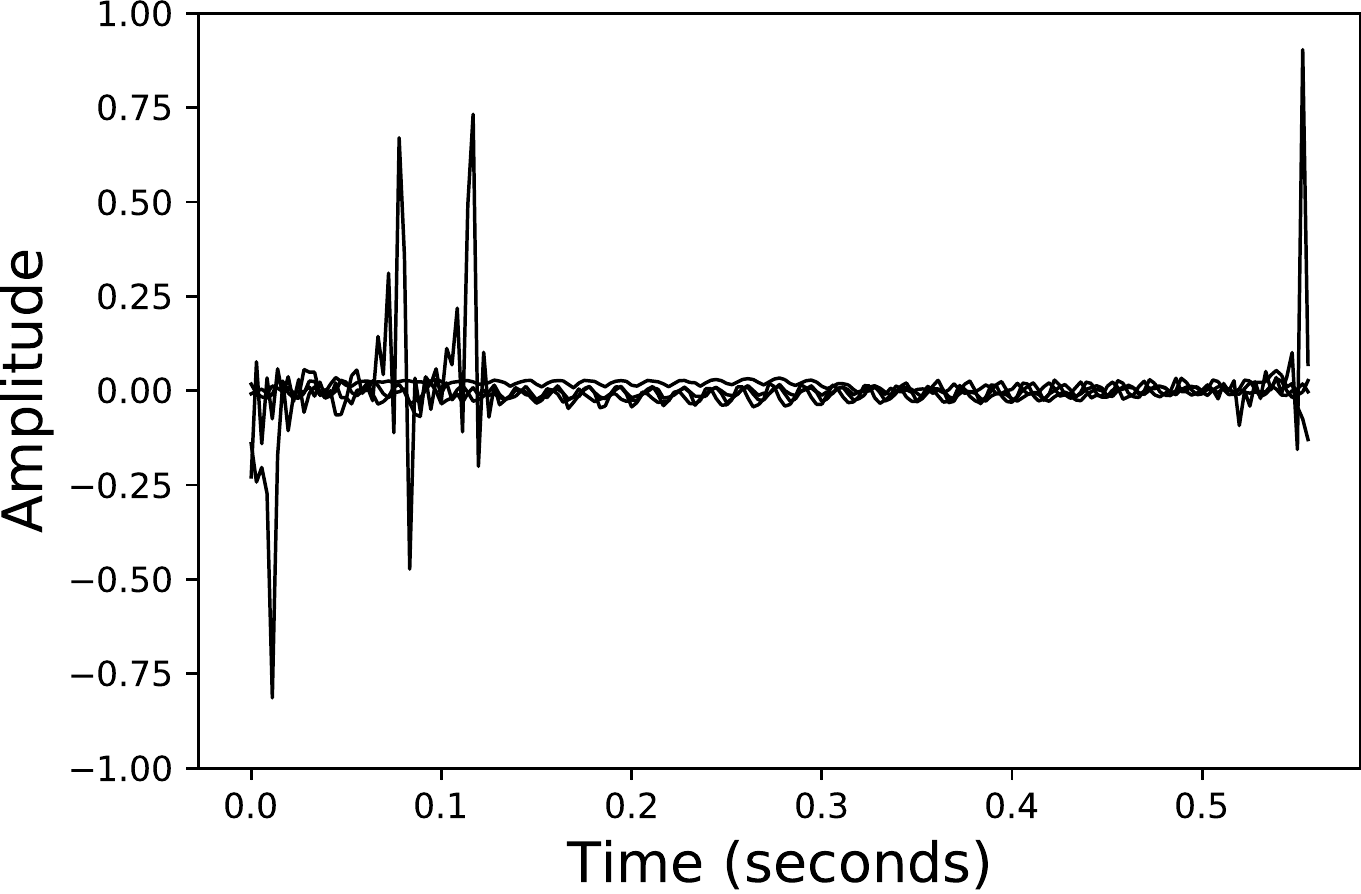}
	\caption{Generators of a continuously shift invariant dictionary learned from ECG data.  We specify as input $q=1$ (top left), $2$ (top right), $3$ (bottom left), $4$ (bottom right) number of generators.}
	\label{fig:CtsShift_ECG}
\end{figure}

\begin{figure}
	\centering
	\includegraphics[width=0.60\textwidth]{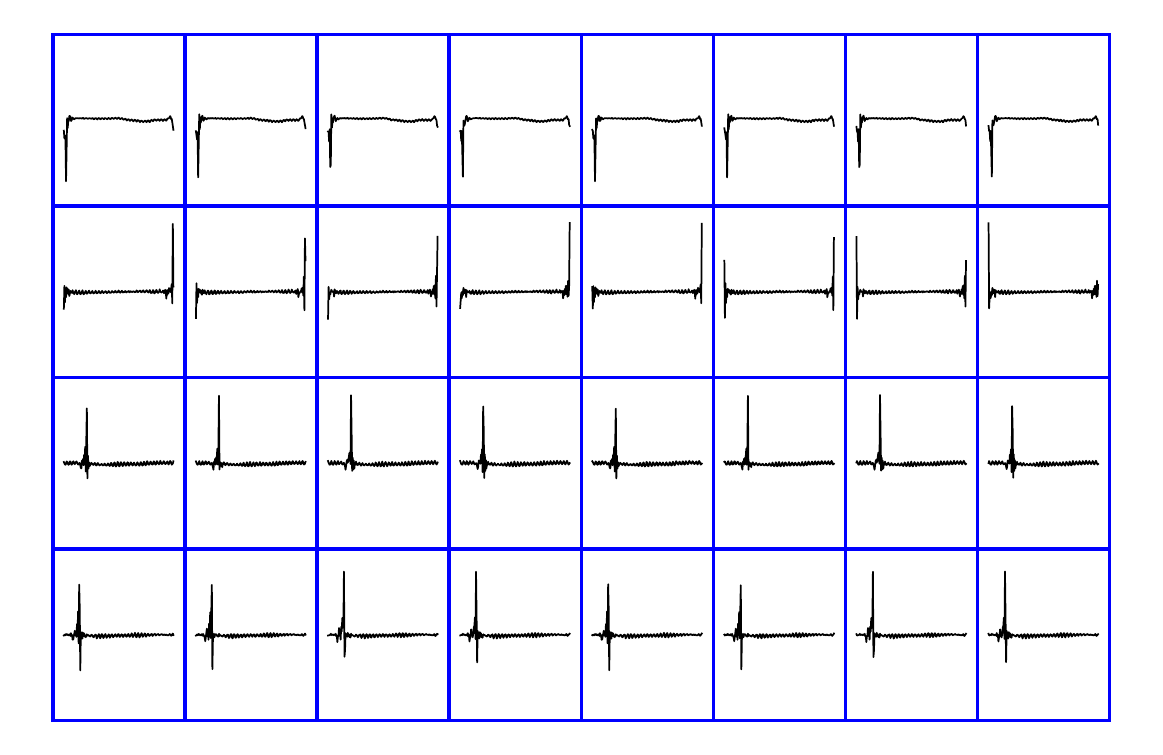}
	\caption{Continuous shift of learned generators.  Each row represents a single generator, and span a single integer coordinate shift.}
	\label{fig:CtsShift_Range}
\end{figure}

\begin{figure}
	\centering
	\includegraphics[width=0.45\textwidth]{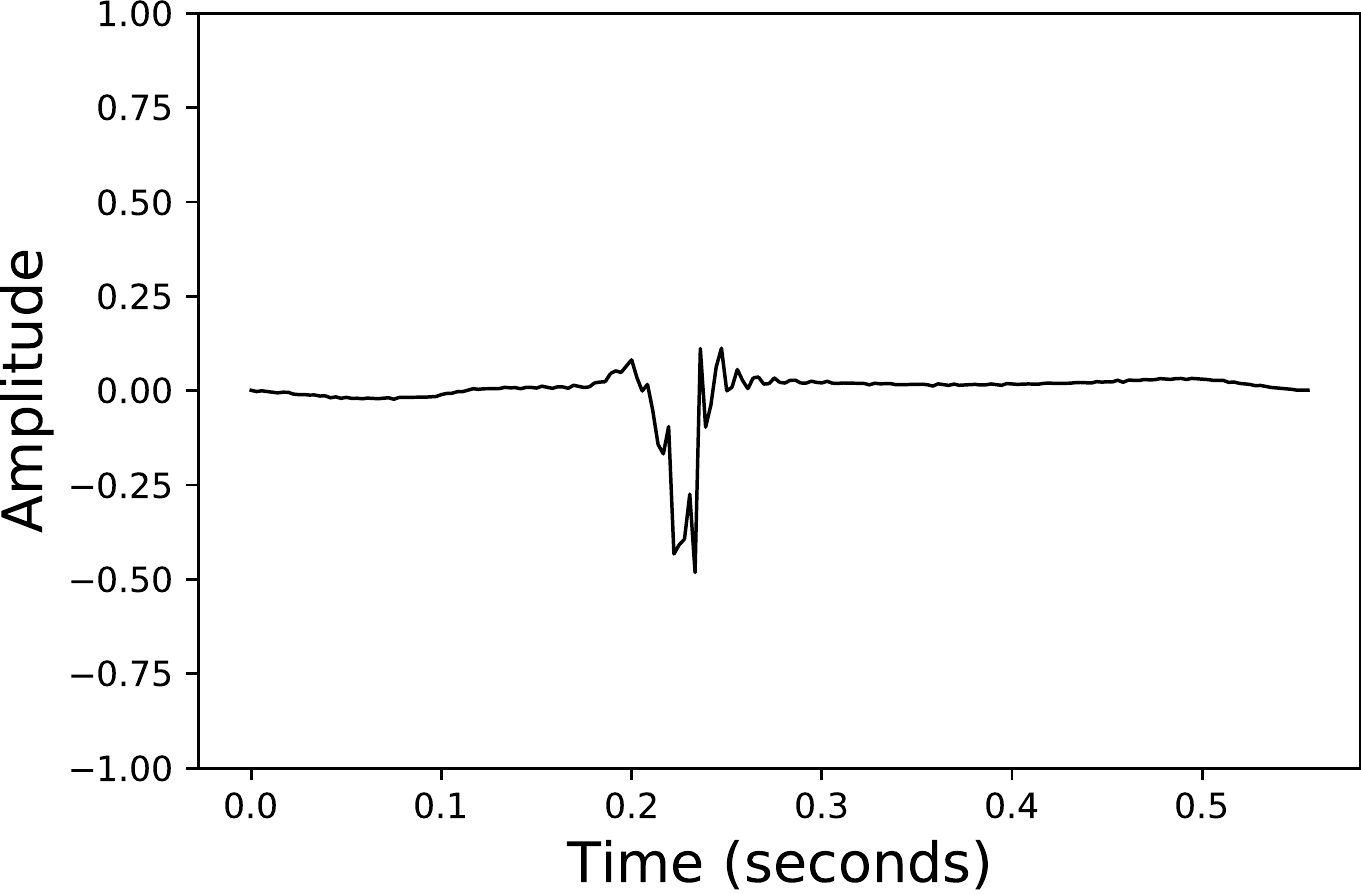}
	\includegraphics[width=0.45\textwidth]{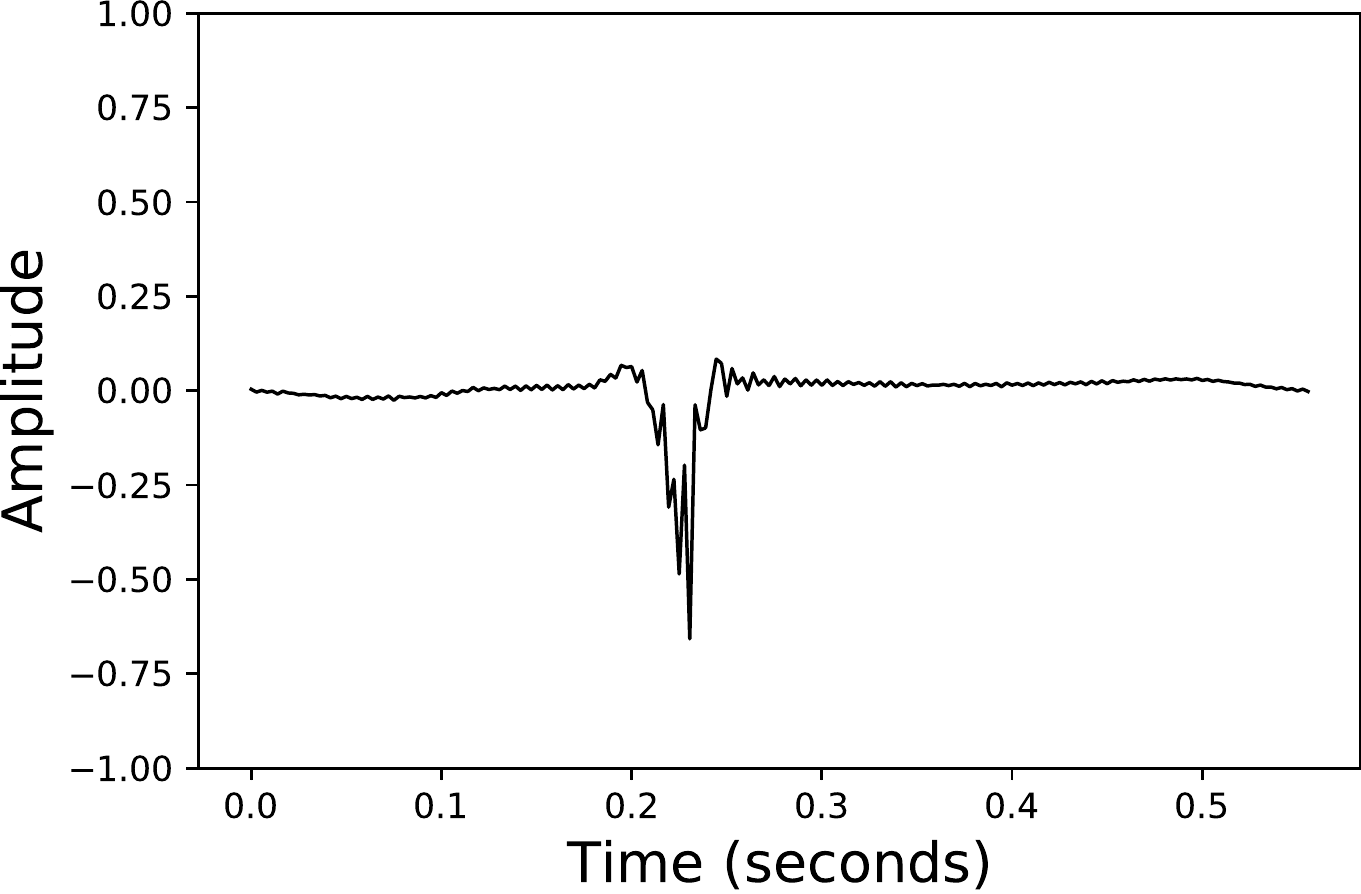}
	\includegraphics[width=0.45\textwidth]{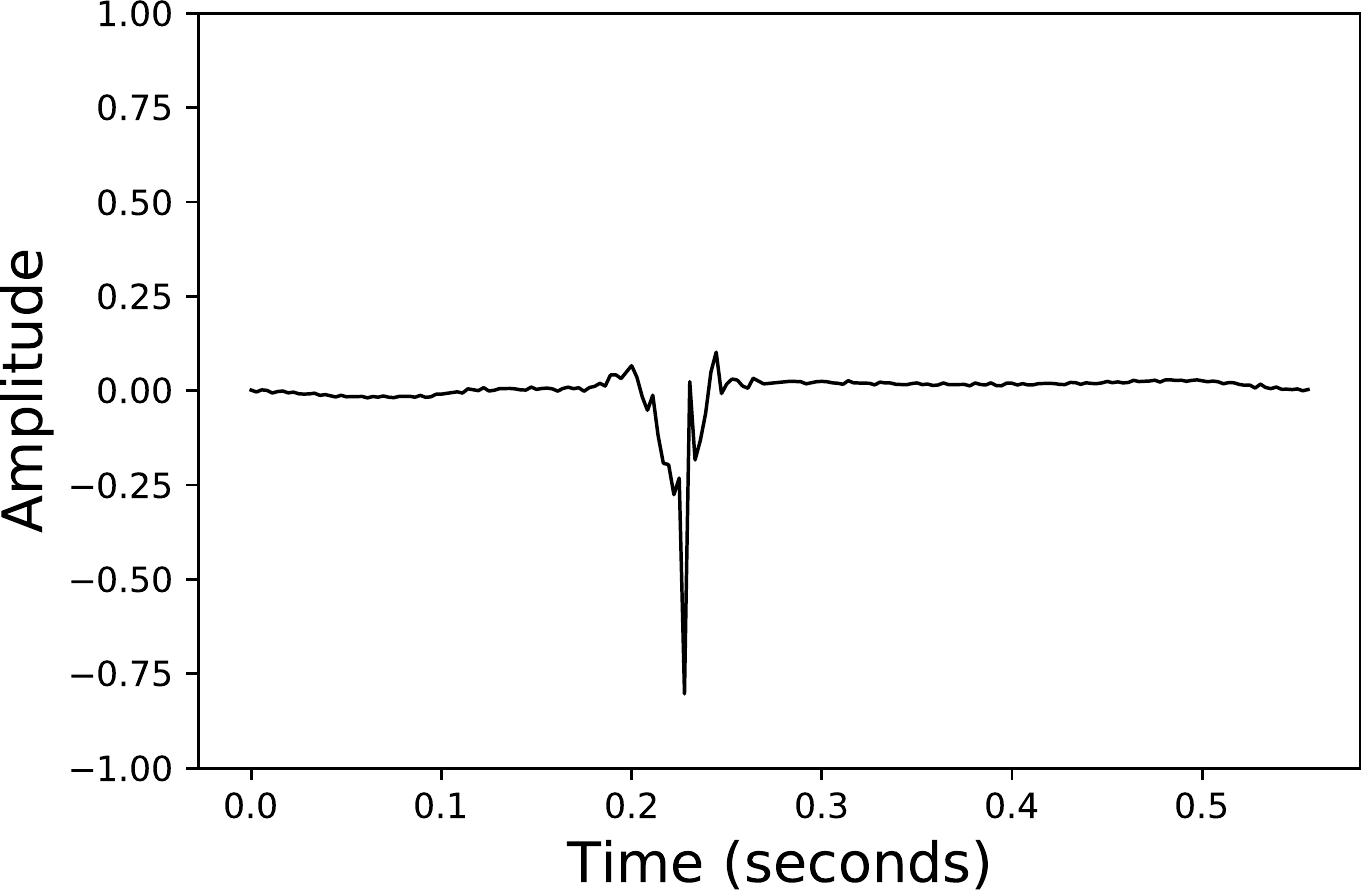}
	\includegraphics[width=0.45\textwidth]{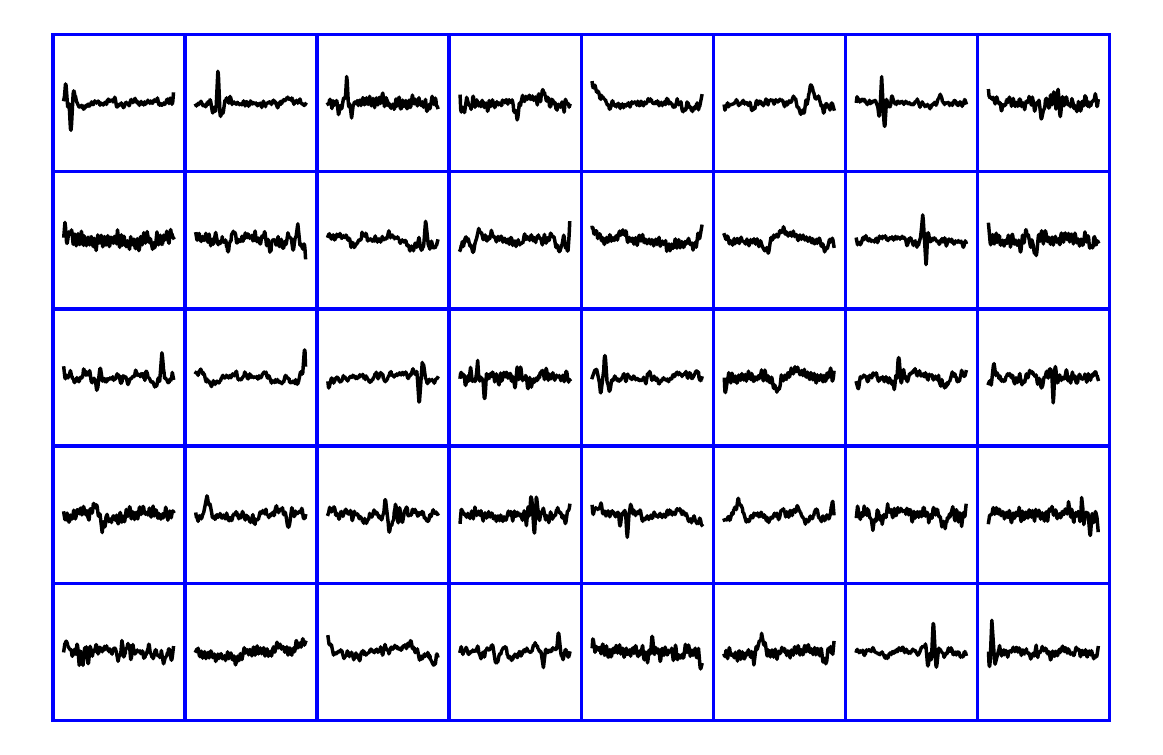}
	\caption{Generator of an integer shift invariant dictionary (top left), interpolated integer shift invariant dictionary ($2\times$ dict. elements) (top right), interpolated integer shift invariant dictionary ($4\times$ dict. elements) (bottom left), and generators from regular dictionary learning (bottom right) learned from ECG data.}
	\label{fig:ECG_othermethods}
\end{figure}

\begin{figure}
\centering
\begin{tabular}{|c|c|}
\hline 
Approach & Avg. per iteration time \\
\hline 
Cts. Shift Inv. DL & 280 secs \\
Convolutional DL (CDL) & 2.2 secs \\
Interpolated CDL ($2\times$ pts.) & 4.2 secs \\
Interpolated CDL ($4\times$ pts.) & 8.3 secs \\
Regular DL & 2.5 secs \\
\hline
\end{tabular}
\caption{Comparison of per iteration time.}
	\label{fig:Iteration_time}
\end{figure}
\subsection{Processing on Unseen Orientations} \label{sec:unseen}

In the following, we consider a task that highlight the utility of expressing the full range of orientations.  In this experimental set-up, our dataset comprises $n=1000$ time series of length $31$.  These signals are segmented from the same ECG signal as in the experimental set-up in Section \ref{sec:ecg_ctsshift} -- the difference is that the signal is sampled at $36$Hz, and the time series only attain a maximum in the first $10$ coordinates (see Figure \ref{fig:LHSdata} for a subset of the data).  Stated simply, the dataset is constructed such that we do not observe the full spectrum of shifts in our data.

\begin{figure}
	\centering
	\includegraphics[width=0.45\textwidth]{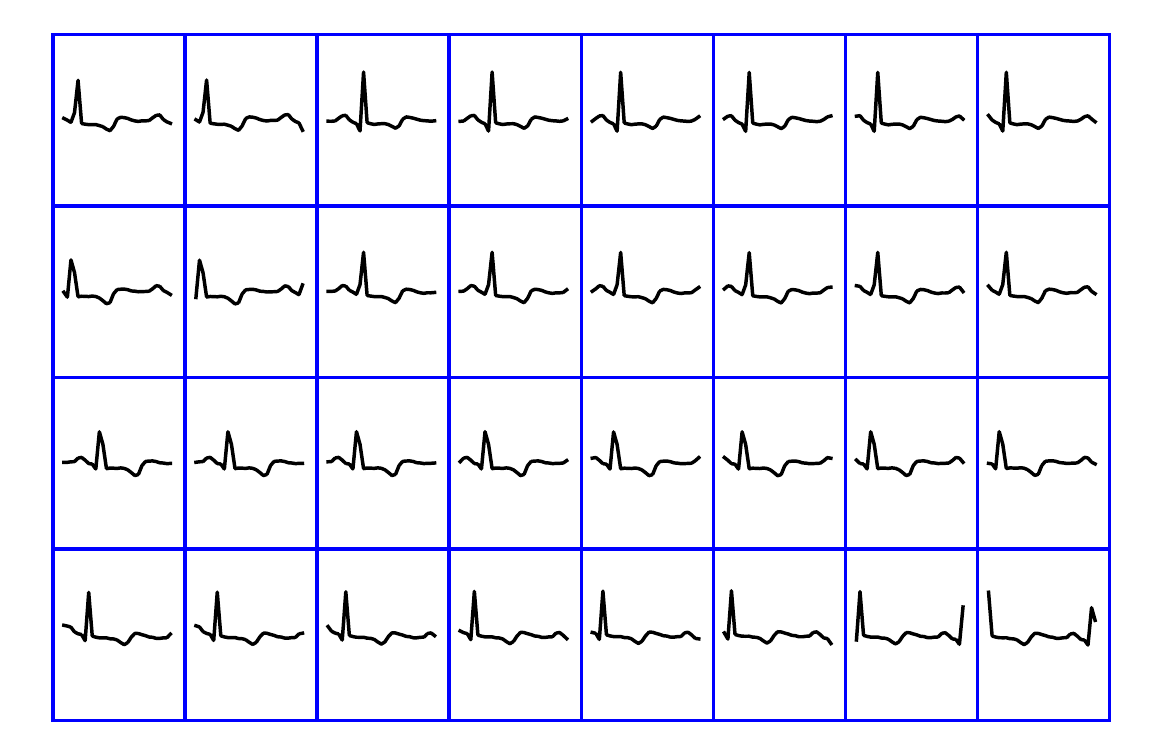}
	\includegraphics[width=0.45\textwidth]{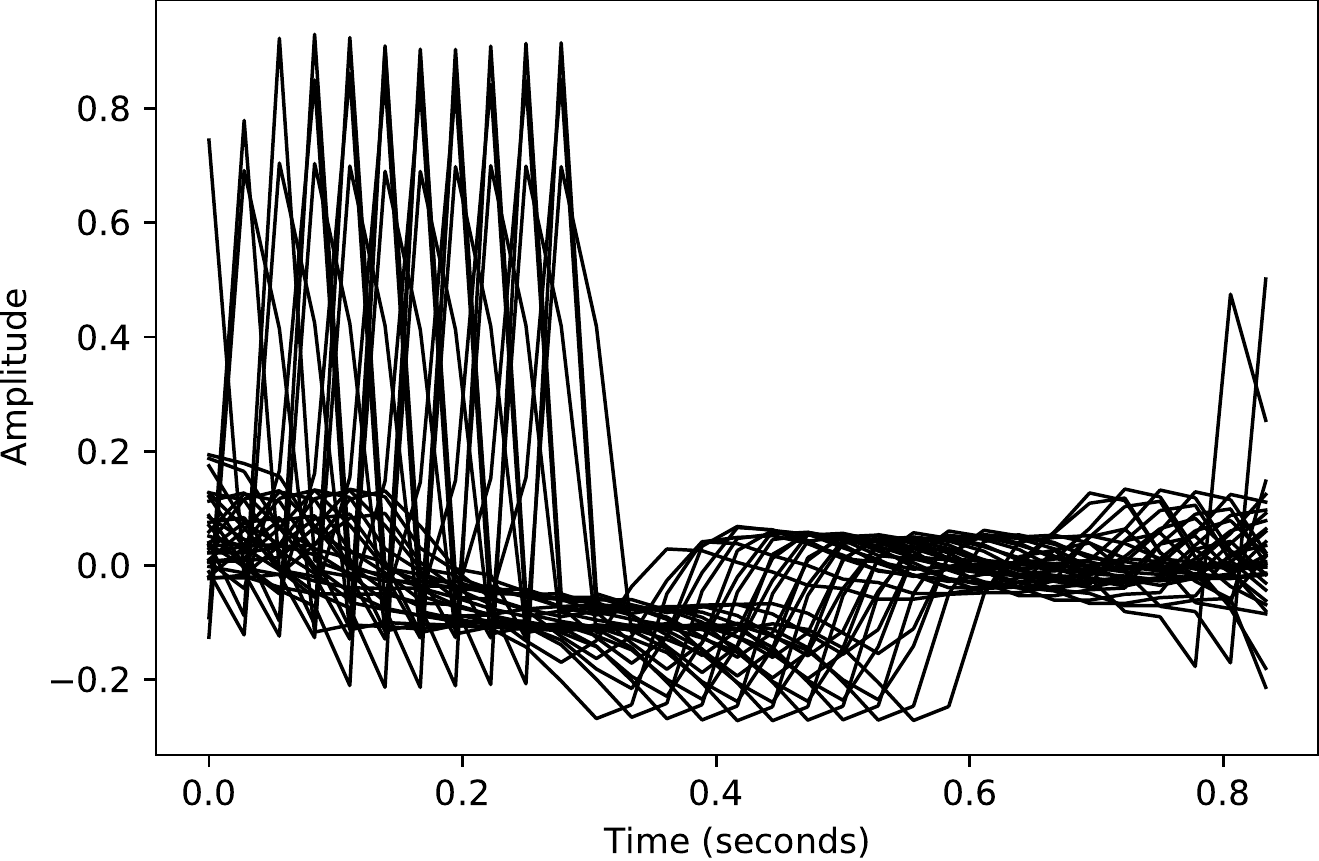}
	\caption{Subset of the dataset presented separately (left sub-plot) and superimposed on the same plot (right sub-plot).}
	\label{fig:LHSdata}
\end{figure}

Our first remark is that regular dictionary learning, when applied to the dataset, does not learn atoms that capture shifts of the data beyond those observed in the dataset.  Figure \ref{fig:CDL_LHSatoms} shows the output by applying regular dictionary learning to learn a dictionary comprising $31$ atoms.  In contrast, our method when applied to the dataset with the choice of a single generator learns a waveform that captures the signal (see Figure \ref{fig:GIDL_LHSatoms}).  The unbalanced nature of the data in the sense that only a fraction of the full spectrum of orientations is represented in the data poses no difficulty to our framework.

\begin{figure}
	\centering
	\includegraphics[width=0.45\textwidth]{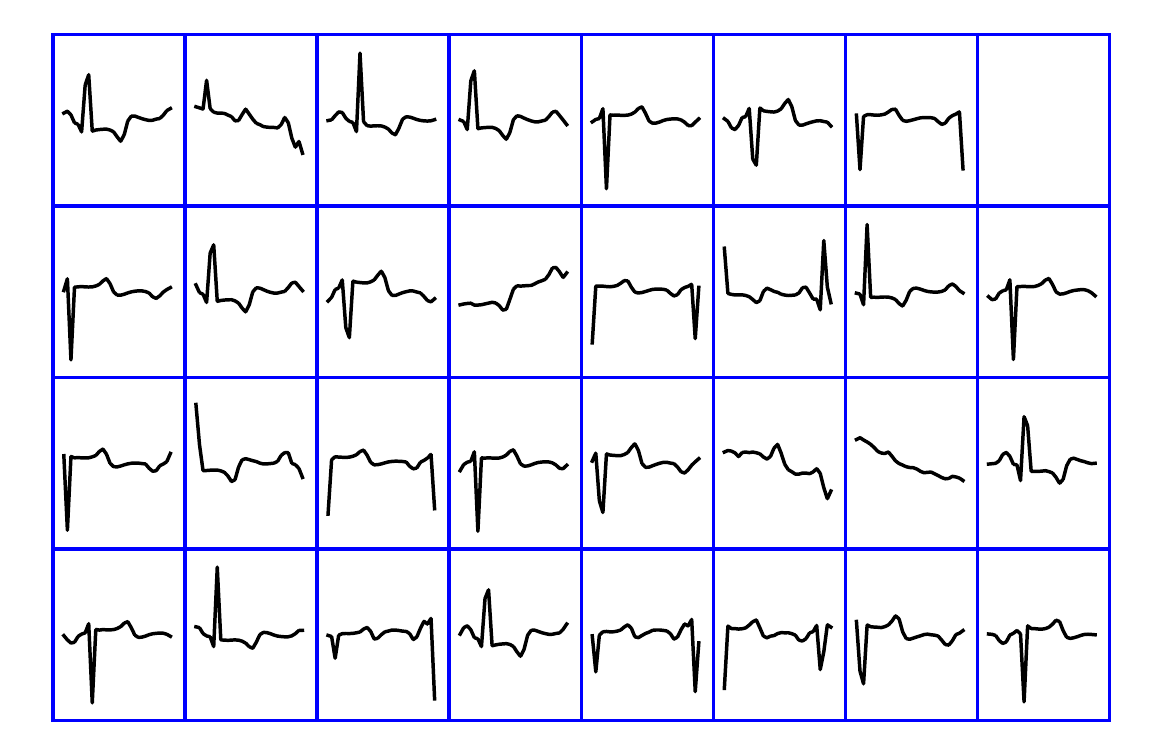}
	\includegraphics[width=0.45\textwidth]{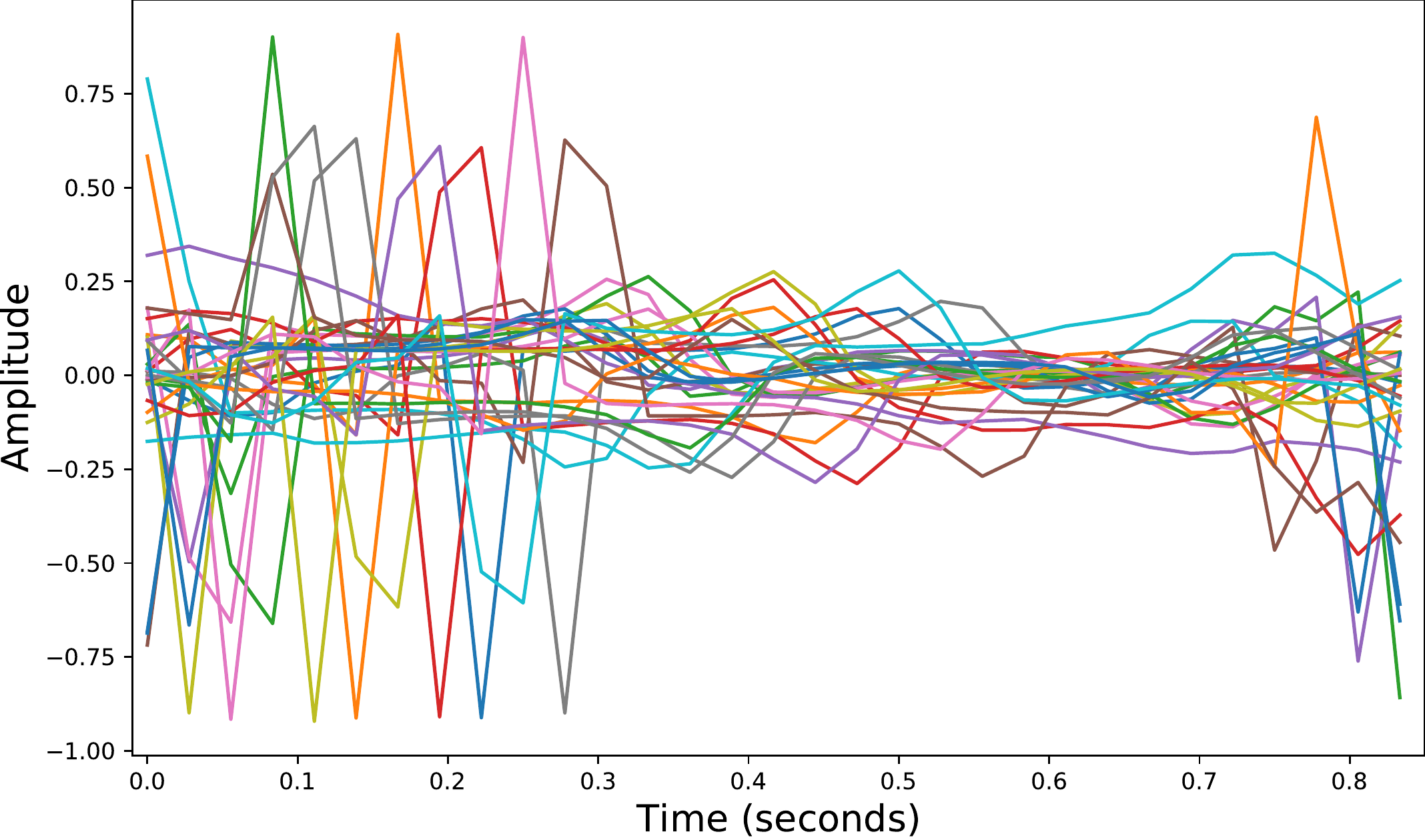}
	\caption{Atoms learned from regular dictionary learning presented separately (left sub-plot) and superimposed on the same plot (right sub-plot).}
	\label{fig:CDL_LHSatoms}
\end{figure}

\begin{figure}
	\centering
	\includegraphics[width=0.5\textwidth]{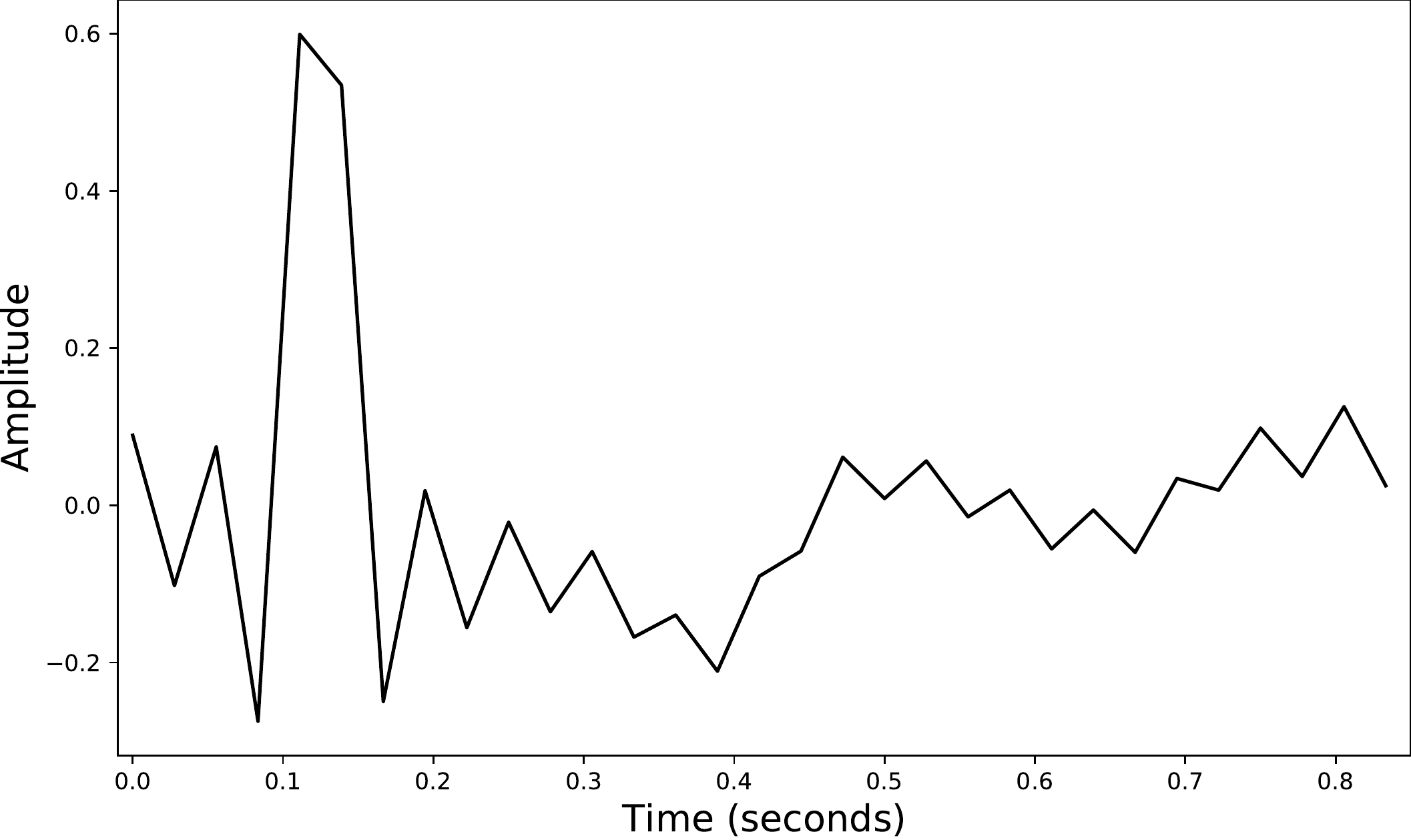}
	\caption{Atom learned from our framework.}
	\label{fig:GIDL_LHSatoms}
\end{figure}

The incorporation of invariance priors becomes particularly useful when we wish to use the learned atoms for processing orientations of data not observed in the training set.  Consider the following instance in which we observe a segment of the waveform -- see left sub-plot of Figure \ref{fig:completion}, and we wish to complete the missing entries.  We do so by seeking the vector that minimizes the norm induced by the learned atoms:
\begin{equation} \label{eq:cvx_completion}
\by_{\mathrm{opt}} ~ \in ~ \underset{\bx}{\arg \min} ~~ \| \bx \| \quad \text{s.t.} \quad P_{\mathrm{obs.}} (\bx - \by_{\mathrm{data}}) = 0.
\end{equation}
Subsequently, the vector $\by_{\mathrm{opt}}$ is the solution of a convex program.  In our numerical experiments, we compute $\by_{\mathrm{opt}}$ using a CVXPY \cite{agrawal2018rewriting,diamond2016cvxpy} implementation of CVXOPT \cite{CVXOPT}.

Figure \ref{fig:completion} shows an example of the completed signal by instantiating $\| \bx \|$ in the above using the norms induced by atoms learned using regular dictionary learning and our method.  We repeat \eqref{eq:cvx_completion} over $100$ different time series $\by_{\mathrm{data}}^{(i)}$, $1 \leq i \leq 100$.  The average squared error loss using our method $\frac{1}{100} \sum_{i=1}^{100} \| \by_{\mathrm{opt}} - \by_{\mathrm{data}} \|_2^2 / \|\by_{\mathrm{data}}\|^2_2 $ is $ 0.079$ using our method, and is $350$ using regular dictionary learning.

\begin{figure}
	\centering
	\includegraphics[width=0.45\textwidth]{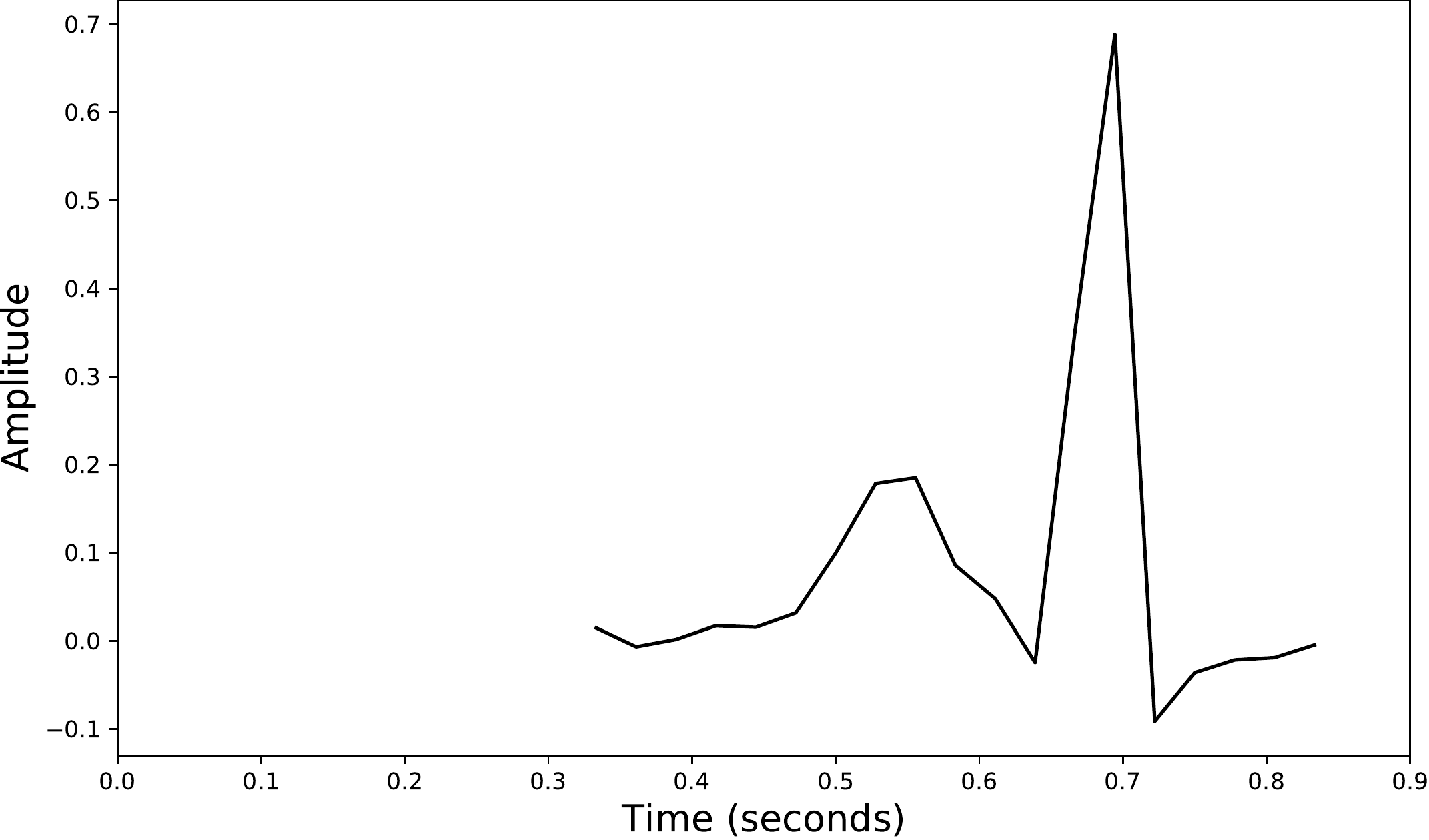}
	\includegraphics[width=0.45\textwidth]{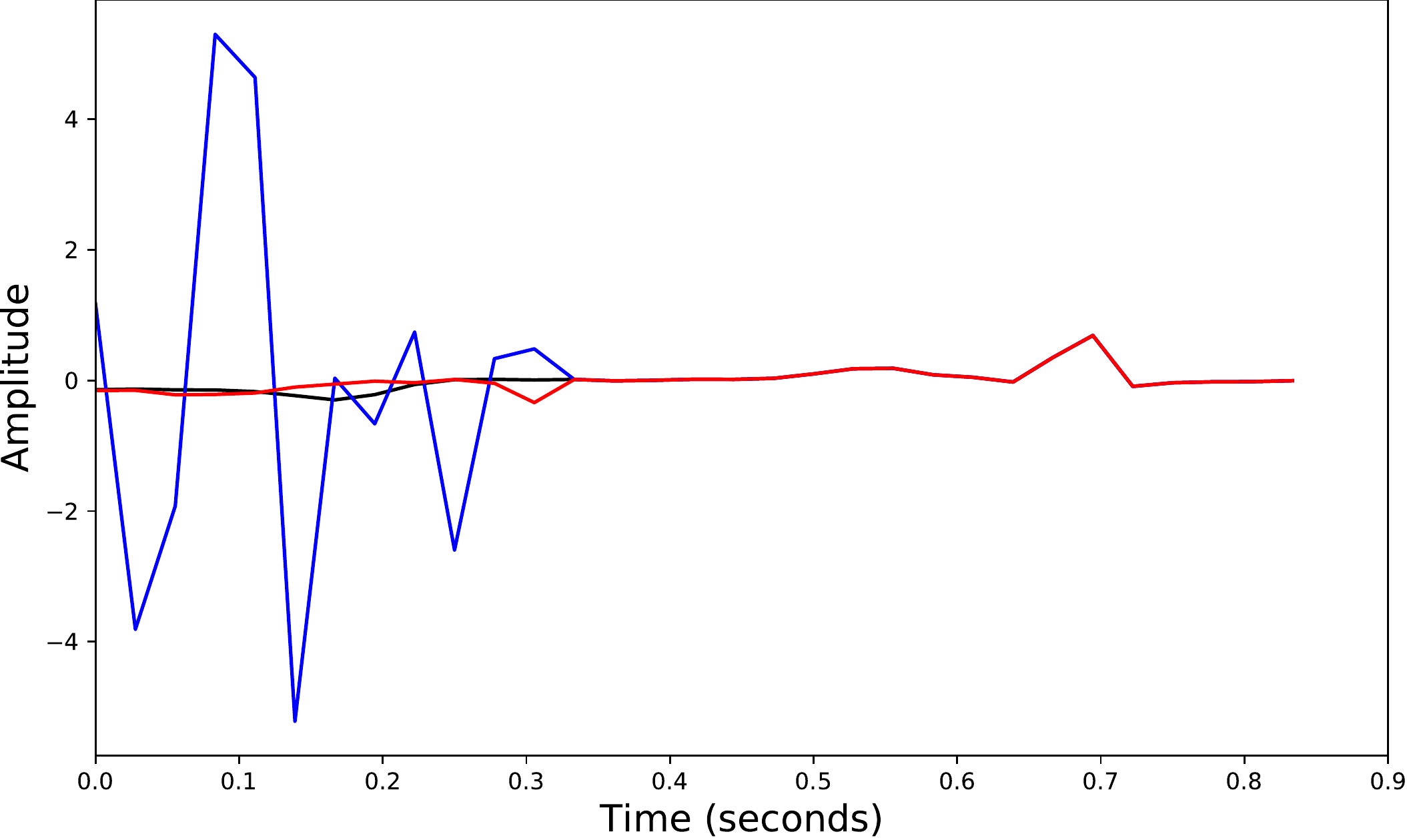}
	\includegraphics[width=0.45\textwidth]{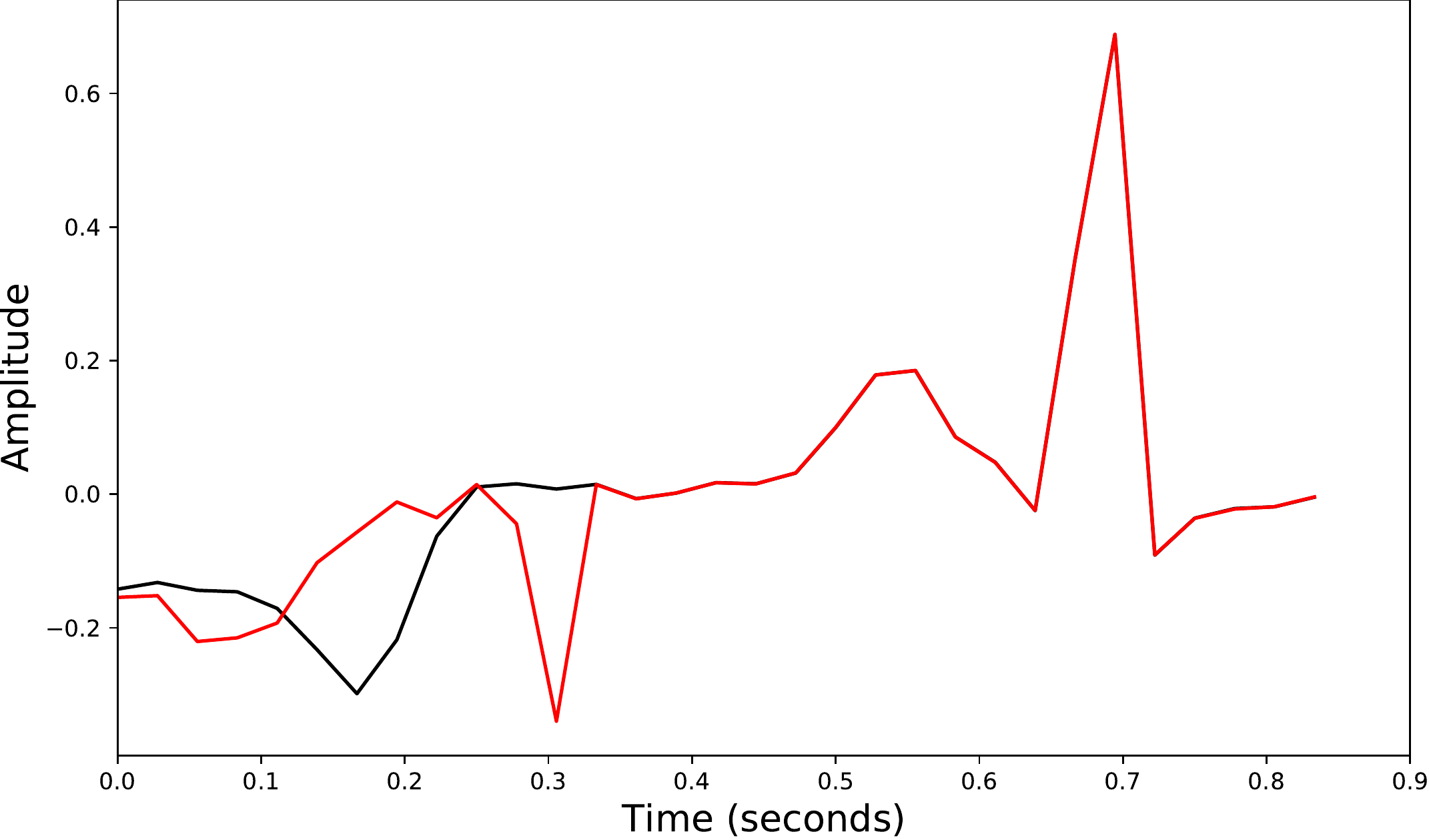}
	\caption{Completing a partially observed ECG signal (top left) using regular dictionary learning (in blue, top right), and our framework (in red, bottom).}
	\label{fig:completion}
\end{figure}
\subsection{Dictionary Learning for Synchronization}

We consider a stylized numerical experiment motivated by applications in synchronization problems.  Suppose we have a collection of $r$ devices that are known to be synchronized relative to each other, but the common phase of these devices is unknown.  Our goal is to estimate the phase given noisy measurements of these objects.

Concretely, consider \emph{matrices} of the form $\{Y^{(i)}\}_{i=1}^{n} \subset \mathbb{R}^{d \times r}$, where $Y^{(i)} = G^{(i)} A^{\star} + E^{(i)}$ is the product of an unknown \emph{orthogonal} matrix $G^{(i)} \in O(d)$ and an unknown common linear map $A^{\star}$, corrupted by noise $E^{(i)}$.  We can view the columns of $Y^{(i)}$ as the noisy measurements of $r$ different devices.  The columns of $A^{\star}$ captures the phase of each device relative to other devices, while the matrix $G^{(i)}$ captures a phase shift that is common to all objects, but specific to an observation instance $i$.

In our first experiment, we generate $n=1000$ data points with the choice of $d=3$ and $r=20$.  We generate $A^{\star}$ from the normal distribution, and we normalize the columns of $A^{\star}$ to be one.  The orthogonal matrices $G^{(i)}$ are drawn from the uniform measure over $O(n)$, and the noise $E^{(i)}$ is drawn from the normal distribution $E^{(i)} \sim \mathcal{N}(0,\sigma I)$, with $\sigma = 0.1$. 

We apply our framework in Section \ref{sec:example_orthogonal} (with the choice of a single generator) to estimate the linear map $A^{\star}$ (we also normalize the columns of each iterate to be unit-norm).  In the left sub-plot of Figure \ref{fig:Sync_DL}, we show the difference between the iterates and the ground truth $A^{\star}$ using the distance measure \eqref{eq:dictdistance}.  Our results show that our method recovers the underlying $A^{\star}$ in all $10$ random initializations.  Note that to compute the distance measure \eqref{eq:dictdistance}, it is necessary to solve an optimization instance of the following form as a sub-routine
\begin{equation*}
\underset{Q \in O(n)}{\arg \min} ~ \| A_1 - Q A_2 \|_F^2.
\end{equation*}
By eliminating the constant terms, this is equivalent to the following
\begin{equation*}
\underset{Q \in O(n)}{\arg \max} ~ \mathrm{trace}(A_2 A_1^{\intercal} Q ).
\end{equation*}
We compute the solution to the above by first computing the SVD $U \Sigma V^{\intercal} = A_1 A_2^{\intercal}$, and setting $Q = UV^{\intercal}$.

\begin{figure}
	\centering
	\includegraphics[width=0.46\textwidth]{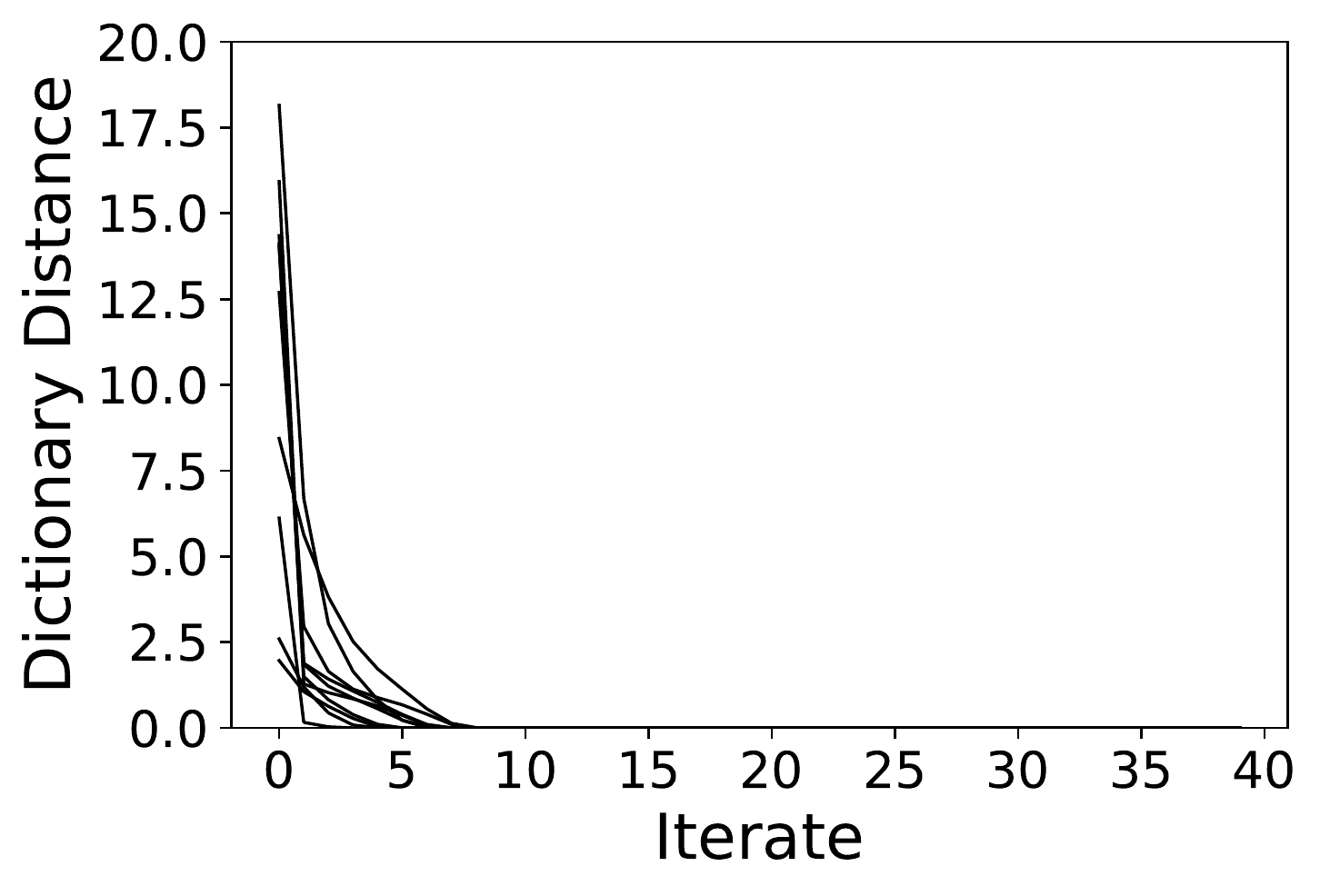}
	\includegraphics[width=0.45\textwidth]{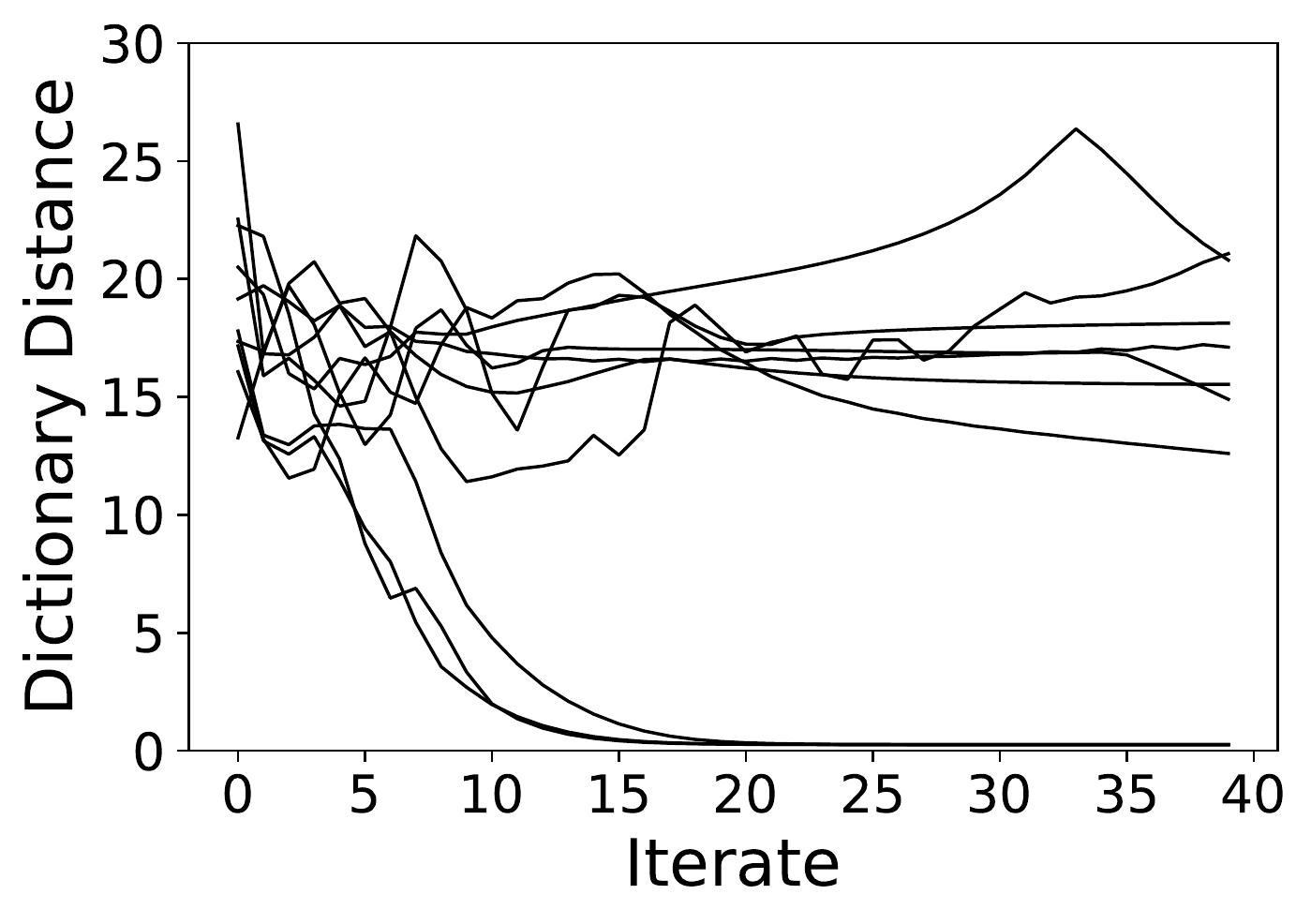}
	\caption{DL invariant to orthogonal transformations.  The graphs show the distance of each iterate from the ground truth dictionary.  The data is synthetically generated by one generator in the left sub-plot and two generators in the right sub-plot.}
	\label{fig:Sync_DL}
\end{figure}

In our second experiment, the signals $Y^{(i)} = \sum_j^{q} c^{(i)}_j G^{(i)}_j A^{\star}_j + E^{(i)}$ are \emph{linear sums} of rotated $A^{\star}_j$'s.  Although the task of recovering the generators $\{A^{\star}_j\}_{j=1}^{q}$ bears less resemblance to synchronization applications, it can be solved using our framework in Section \ref{sec:example_orthogonal}.  We generate $n=1000$ data matrices with the same choice of parameters $d=3$, $r=20$, and $\sigma = 0.1$.  Here, the additional scaling variable $c^{(i)}_j \sim \mathcal{N}(0,1)$ is standard normal.  We show the corresponding set of results in the right sub-plot of Figure \ref{fig:Sync_DL}.  In this instance, our method is less successful because of the increase in number of unknown variables; our method appears to succeed only in $3$ out of $10$ random initializations.

In summary, this experiment describes a concrete instance where the type of symmetry (namely the orthogonal group) necessitates the flexibility that our framework provides: With continuous shifts, one could apply discrete analogs as an approximation whereas with orthogonal rotations, there is no obvious discrete analog.  From a computational perspective, the instantiation of our framework for this example is a relative simple procedure in that it only requires computing the SVDs of certain matrices as a computational primitive. 

\section{Conclusion and Future Directions} \label{sec:conc}

In this paper, we develop a framework for representing data as sparse linear sums drawn from a collection of basic building blocks subject to the constraint that the collection respects a pre-specified symmetry.  Our results show that the incorporation of such symmetries as priors is most useful when dataset has few data-points, and when the full spectrum of symmetries is inadequately represented in the dataset.  

In the following, we discuss some future directions based on our work.

\textbf{Misspecification, subgroups, and convex relaxations.}  Our first future direction concerns studying the robustness of our framework to misspecification of the transformation group.  One instance in which misspecification arises is when we wish to specify a sub-group rather than the full spectrum of symmetries in the data, e.g., discrete shifts as opposed to continuous shifts.  We may do so, for instance, out of computational considerations.

We note that the atomic norms induced by the original transformation group and the sub-group are closely related -- specifically, the level set of the atomic norm induced by the sub-group of transformations would be a convex \emph{inner approximation} of the level set of the original atomic norm.   As such, we have a strong reason to expect that atoms learned using both sets of transformation groups will be closely related.  In fact, convex inner (and outer) approximations are prominently used to develop more tractable approximations of convex programs.  For such reasons, it would be useful to understand how the learned atoms using both sets of approaches differ.

\textbf{Specifying subsets of the full invariance group and the associated computational and statistical tradeoffs.}  One interesting conceptual question that warrants deeper investigation: Given a processing task involving data that possesses some form of invariance, is it always advantageous to incorporate the full range of symmetries into the learning task?  Our numerical experiments in Sections \ref{sec:shiftinvariantdiscrete} and \ref{sec:unseen} emphasize the importance of incorporating some basic level of invariances.  On the other hand, our experiments in Section \ref{sec:shiftinvariantdiscrete} also suggest that it may not be necessary to incorporate the full extent of symmetries.  In fact, it may very well be possible to learn representations of comparable quality by only incorporating some of the symmetries to those learned by incorporating the full spectrum of symmetries.  This insight may be useful from a computational perspective because the former class of methods may be cheaper.  It would be useful, from a conceptual as well as a practical perspective, to understand the computational and statistical tradeoffs associated with specifying only part of the full range of symmetries in a generic representation learning task.

\textbf{Applications to neural networks.}  In the Introduction, we discussed a closely related body of work extending CNNs to more general invariances \cite{WGTB:17,GD:14,CW:16}.  As we noted, the proposed neural network architectures do not incorporate sparse structure whereas learning sparse representations is central to our set-up.  There are certain applications where incorporating sparse priors within the neural network architecture is beneficial, and it would be interesting to see if the framework we provide via atomic norms suggests natural analogs for neural networks. 

\textbf{Generic recipes for obtaining tractable descriptions of the convex hull of matrix groups.}  An important future direction is to develop a procedure for obtaining tractable descriptions of convex hulls for a broad collection of matrix groups.  This procedure will allow us to learn dictionaries that are invariant to larger collections of symmetries.

A class of objects that has been widely studied in convex algebraic geometry and optimization concerns, and is closely related to the above problem is the collection of \emph{orbitopes} \cite{SSS:11}.  An orbitope is defined to be the convex hull of the orbit of a compact algebraic group acting linearly on a vector space \cite{SSS:11}.  The level set of the atomic norm we apply in our framework is therefore an orbitope, provided the matrix group describing the symmetries is compact algebraic.  Furthermore, the matrix group elements can be viewed as orbits of the identity matrix $I$, and hence the convex hull of these elements is also an orbitope.  The algebraic properties, the geometric properties, and the optimization-related aspects of orbitopes are areas of active research interests -- in particular, there is a body of work that concerns providing descriptions of orbitopes via semidefinite programming \cite{SanSau:20,FSP:15,SSS:11}.  It would be useful to build upon these techniques to further understand the types of invariances that are expressible within our framework.

\section*{Acknowledgements}

The author is supported by the Agency for Science, Technology and Research (A*STAR) under its AME Programmatic Funding Scheme (Project No. A18A1b0045) and the Ministry of Education (Singapore) Academic Research Fund (Tier 1) R-146-000-329-133.  The author would like to thank Venkat Chandrasekaran for insightful discussions held during the beginning of this project, and Riley J. Murray for providing helpful directions concerning numerical experiments.  The author also wishes to thank the three reviewers for feedback that has substantially improved the paper.

\bibliography{bib_gidl}

\begin{thebibliography}{10}
\providecommand{\url}[1]{#1}
\csname url@samestyle\endcsname
\providecommand{\newblock}{\relax}
\providecommand{\bibinfo}[2]{#2}
\providecommand{\BIBentrySTDinterwordspacing}{\spaceskip=0pt\relax}
\providecommand{\BIBentryALTinterwordstretchfactor}{4}
\providecommand{\BIBentryALTinterwordspacing}{\spaceskip=\fontdimen2\font plus
\BIBentryALTinterwordstretchfactor\fontdimen3\font minus
  \fontdimen4\font\relax}
\providecommand{\BIBforeignlanguage}[2]{{%
\expandafter\ifx\csname l@#1\endcsname\relax
\typeout{** WARNING: IEEEtran.bst: No hyphenation pattern has been}%
\typeout{** loaded for the language `#1'. Using the pattern for}%
\typeout{** the default language instead.}%
\else
\language=\csname l@#1\endcsname
\fi
#2}}
\providecommand{\BIBdecl}{\relax}
\BIBdecl

\bibitem{OlsFie:96}
B.~A. Olshausen and D.~J. Field, ``{E}mergence of {S}imple-{C}ell {R}eceptive
  {F}ield {P}roperties by {L}earning a {S}parse {C}ode for {N}atural
  {I}mages,'' \emph{Nature}, vol. 381, pp. 607--609, 1996.

\bibitem{OlsFie:97}
------, ``{S}parse {C}oding with an {O}vercomplete {B}asis {S}et: {A}
  {S}trategy {E}mployed by {V}1?'' \emph{Vision Research}, vol.~37, no.~23, pp.
  3311--3325, 1997.

\bibitem{AEB:06}
M.~Aharon, M.~Elad, and A.~Bruckstein, ``{K}-{SVD}: {A}n {A}lgorithm for
  {D}esigning {O}vercomplete {D}ictionaries for {S}parse {R}epresentation,''
  \emph{IEEE Transactions on Signal Processing}, vol.~54, no.~11, pp.
  4311--4322, 2006.

\bibitem{MBP:14}
J.~Mairal, F.~Bach, and J.~Ponce, ``{S}parse {M}odeling for {I}mage and
  {V}ision {P}rocessing,'' \emph{Foundations and Trends in Computer Graphics
  and Vision}, vol.~8, no. 2--3, pp. 85--283, 2014.

\bibitem{Ela:10}
M.~Elad, \emph{{S}parse and {R}edundant {R}epresentations: {F}rom {T}heory to
  {A}pplications in {S}ignal and {I}mage {P}rocessing}.\hskip 1em plus 0.5em
  minus 0.4em\relax Springer, 2010.

\bibitem{ElaAha:06}
M.~Elad and M.~Aharon, ``{I}mage {D}enoising via {S}parse and {R}edundant
  {R}epresentations {O}ver {L}earned {D}ictionaries,'' \emph{IEEE Transactions
  on Signal Processing}, vol.~15, no.~12, pp. 3736--3745, 2006.

\bibitem{Wohlberg:15}
B.~Wohlberg, ``{E}fficient {A}lgorithms for {C}onvolutional {S}parse
  {R}epresentations,'' \emph{IEEE Transactions on Image Processing}, vol.~25,
  no.~1, 2015.

\bibitem{LewSej:17}
M.~S. Lewicki and T.~J. Sejnowski, ``{C}oding {T}ime-varying {S}ignals using
  {S}parse, {S}hift-invariant {R}epresentations,'' in \emph{Proceedings of the
  1998 Conference on Advances in Neural Information Processing Systems II},
  1998.

\bibitem{PRSE:17}
V.~Papyan, Y.~Romano, J.~Sulam, and M.~Elad, ``{C}onvolutional {D}ictionary
  {L}earning via {L}ocal {P}rocessing,'' in \emph{The IEEE International
  Conference on Computer Vision}, 2017.

\bibitem{JVLG:06}
P.~Jost, P.~Vandergheynst, S.~Lesage, and R.~Gribonval, ``{M}o{T}{I}{F} : {A}n
  {E}fficient {A}lgorithm for {L}earning {T}ranslation {I}nvariant
  {D}ictionaries,'' in \emph{Proceedings of the International Conference on
  Acoustics, Speech, and Signal Processing (ICASSP)}, 2006.

\bibitem{BluDav:06}
T.~Blumensath and M.~E. Davies, ``Sparse and {S}hift-{I}nvariant
  {R}epresentations of {M}usic,'' \emph{IEEE Transactions on Audio, Speech, and
  Language Processing}, vol.~14, no.~1, 2006.

\bibitem{PABD:06}
M.~D. Plumbley, S.~A. Abdallah, T.~Blumensath, and M.~E. Davies, ``{S}parse
  {R}epresentations of {P}olyphonic {M}usic,'' \emph{Signal Processing},
  vol.~86, no.~3, 2006.

\bibitem{GRKN:07}
R.~Grosse, R.~Raina, H.~Kwong, and A.~Y. Ng, ``{S}hift-invariant {S}parse
  {C}oding for {A}udio {C}lassification,'' in \emph{Proceedings of the
  Twenty-Third Conference on Uncertainty in Artificial Intelligence}, 2007.

\bibitem{Rusu:20}
C.~Rusu, ``{O}n {L}earning with {S}hift-invariant {S}tructures,'' \emph{Digital
  Signal Processing}, vol.~99, 2020.

\bibitem{ZSE:19}
E.~Zisselman, J.~Sulam, and M.~Elad, ``{A} {L}ocal {B}lock {C}oordinate
  {D}escent {A}lgorithm for the {C}{S}{C} {M}odel,'' in \emph{IEEE/CVF
  Conference on Computer Vision and Pattern Recognition (CVPR)}, 2019.

\bibitem{GW:18}
C.~Garcia-Cardona and B.~Wohlberg, ``{C}onvolutional {D}ictionary {L}earning:
  {A} {C}omparative {R}eview and {N}ew {A}lgorithms,'' \emph{IEEE Transactions
  on Computational Imaging}, vol.~4, no.~3, 2018.

\bibitem{LGWY:18}
J.~Liu, C.~Garcia-Cardona, B.~Wohlberg, and W.~Yin, ``{F}irst- and
  {S}econd-{O}rder {M}ethods for {O}nline {C}onvolutional {D}ictionary
  {L}earning,'' \emph{SIAM Journal on Imaging Sciences}, vol.~11, no.~2, 2018.

\bibitem{KSH:12}
A.~Krizhevsky, I.~Sutskever, and G.~E. Hinton, ``{I}magenet {C}lassification
  with {D}eep {C}onvolutional {N}eural {N}etworks,'' in \emph{Advances in
  Neural Information Processing Systems}, 2012.

\bibitem{WGTB:17}
D.~E. Worrall, S.~J. Garbin, D.~Turmukhambetov, and G.~J. Brostow, ``{H}armonic
  {N}etworks: {D}eep {T}ranslation and {R}otation {E}quivariance,'' in
  \emph{The IEEE Conference on Computer Vision and Pattern Recognition}, 2017.

\bibitem{GD:14}
R.~Gens and P.~M. Domingos, ``{D}eep {S}ymmetry {N}etworks,'' in \emph{Advances
  in Neural Information Processing Systems}, 2014.

\bibitem{CW:16}
T.~Cohen and M.~Welling, ``{G}roup {E}quivariant {C}onvolutional {N}etworks,''
  in \emph{Proceedings of The 33rd International Conference on Machine
  Learning}, 2016.

\bibitem{KT:18}
R.~Kondor and S.~Trivedi, ``On the generalization of equivariance and
  convolution in neural networks to the action of compact groups,'' in
  \emph{International Conference on Machine Learning (ICML)}, 2018.

\bibitem{SC:19}
Y.~S. Soh and V.~Chandrasekaran, ``{L}earning {S}emidefinite {R}egularizers,''
  \emph{Foundations of Computational Mathematics}, vol.~19, 2019.

\bibitem{CRPW:12}
V.~Chandrasekaran, B.~Recht, P.~A. Parrilo, and A.~S. Willsky, ``{T}he {C}onvex
  {G}eometry of {L}inear {I}nverse {P}roblems,'' \emph{Foundations of
  Computational Mathematics}, vol.~12, no.~6, pp. 805--849, 2012.

\bibitem{Don:06}
D.~L. Donoho, ``{C}ompressed {S}ensing,'' \emph{IEEE Transactions on
  Information Theory}, vol.~52, no.~4, pp. 1289--1306, 2006.

\bibitem{Don:06b}
------, ``{F}or {M}ost {L}arge {U}nderdetermined {S}ystems of {L}inear
  {E}quations the {M}inimal $\ell_1$-norm {S}olution {I}s {A}lso the {S}parsest
  {S}olution,'' \emph{Communications on Pure and Applied Mathematics}, vol.~59,
  no.~6, pp. 797--829, 2006.

\bibitem{CRT:06}
E.~J. Cand\`es, J.~Romberg, and T.~Tao, ``Robust {U}ncertainty {P}rinciples:
  {E}xact {S}ignal {R}econstruction from {H}ighly {I}ncomplete {F}requency
  {I}nformation,'' \emph{IEEE Transactions on Information Theory}, vol.~52,
  no.~2, pp. 489--509, 2006.

\bibitem{Donoho:95}
D.~L. Donoho, ``{D}e-noising by {S}oft-{T}hresholding,'' \emph{IEEE
  Transactions on Information Theory}, vol.~31, no.~3, 1995.

\bibitem{DonJoh:94}
D.~L. Donoho and I.~M. Johnstone, ``{I}deal {S}patial {A}daptation by {W}avelet
  {S}hrinkage,'' \emph{Biometrika}, vol.~81, no.~3, 1994.

\bibitem{BTR:13}
B.~N. Bhaskar, G.~Tang, and B.~Recht, ``{A}tomic {N}orm {D}enoising with
  {A}pplications to {L}ine {S}pectral {E}stimation,'' \emph{IEEE Transactions
  on Signal Processing}, vol.~61, no.~23, pp. 5987--5999, 2013.

\bibitem{EAH:99}
K.~Engan, S.~O. Aase, and J.~H. Husoy, ``Method of {O}ptimal {D}irections for
  {F}rame {D}esign,'' in \emph{IEEE International Conference on Acoustics,
  Speech, and Signal Processing}, 1999.

\bibitem{Sch:16}
K.~Schnass, ``{C}onvergence {R}adius and {S}ample {C}omplexity of {I}{T}{K}{M}
  {A}lgorithms for {D}ictionary {L}earning,'' \emph{Applied and Computational
  Harmonic Analysis}, 2016.

\bibitem{CDS:98}
S.~S. Chen, D.~L. Donoho, and M.~A. Saunders, ``{A}tomic {D}ecomposition by
  {B}asis {P}ursuit,'' \emph{SIAM Journal on Scientific Computing}, vol.~20,
  no.~1, pp. 33--61, 1998.

\bibitem{SFB:19}
A.~H. Song, F.~J. Flores, and D.~Ba, ``{C}onvolutional {D}ictionary {L}earning
  with {G}rid {R}efinement,'' \emph{IEEE Transactions on Signal Processing},
  vol.~68, pp. 2558--2573, 2020.

\bibitem{TBSR:13}
G.~Tang, B.~N. Bhaskar, P.~Shah, and B.~Recht, ``{C}ompressed {S}ensing {O}ff
  the {G}rid,'' \emph{IEEE Transactions on Information Theory}, vol.~59,
  no.~11, 2013.

\bibitem{CarFej:11}
C.~Carath\'eodory and L.~Fej\'er, ``\"{U}ber den zusammenhang der extremen von
  harmonischen funktionen mit ihren koeffizienten und \"uber den
  picard-landauschen satz,'' \emph{Rendiconti del Circolo Matematico di
  Palermo}, vol.~32, no.~1, 1911.

\bibitem{Caratheodory:11}
C.~Carath\'eodory, ``\"{U}ber den variabilit\"atsbereich der fourierschen
  konstanten von positiven harmonischen funktionen,'' \emph{Rendiconti del
  Circolo Matematico di Palermo}, vol.~32, no.~1, 1911.

\bibitem{Toeplitz:11}
O.~Toeplitz, ``Zur theorie der quadratischen und bilinearen formen von
  unendlichvielen ver\"anderlichen,'' \emph{Mathematische Annalen}, vol.~70,
  no.~3, 1911.

\bibitem{Ren:01}
J.~Renegar, \emph{{A} {M}athematical {V}iew of {I}nterior-{P}oint {M}ethods in
  {C}onvex {O}ptimization}.\hskip 1em plus 0.5em minus 0.4em\relax MOS-SIAM
  Series on Optimization, 2001.

\bibitem{NesNem:94}
Y.~Nesterov and A.~Nemirovskii, \emph{{I}nterior-{P}oint {P}olynomial
  {A}lgorithms in {C}onvex {P}rogramming}.\hskip 1em plus 0.5em minus
  0.4em\relax SIAM Studies in Applied and Numerical Mathematics, 1994.

\bibitem{BD:86}
J.~P. Boyle and R.~L. Dykstra, ``{A} {M}ethod for {F}inding {P}rojections onto
  the {I}ntersection of {C}onvex {S}ets in {H}ilbert {S}paces,'' in
  \emph{Advances in Order Restricted Statistical Inference, Lecture Notes in
  Statistics}, R.~L. Dykstra, T.~Robertson, and T.~T. Wright, Eds.\hskip 1em
  plus 0.5em minus 0.4em\relax Springer, New York, 1986, pp. 28--47.

\bibitem{SufHay:93}
T.~J. Suffridge and T.~L. Hayden, ``{A}pproximation by a {H}ermitian {P}ositive
  {S}emidefinite {T}oeplitz {M}atrix,'' \emph{SIAM Journal on Matrix Analysis
  and Applications}, vol.~14, no.~3, 1993.

\bibitem{Higham:88}
N.~J. Higham, ``{C}omputing a {N}earest {S}ymmetric {P}ositive {S}emidefinite
  {M}atrix,'' \emph{Linear Algebra and its Applications}, vol. 103, 1988.

\bibitem{mitbih-physionet:01}
R.~G. Mark and G.~B. Moody, ``{T}he impact of the {M}{I}{T}-{B}{I}{H}
  {A}rrhythmia {D}atabase,'' \emph{IEEE Engineering in Medicine and Biology
  Magazine}, vol.~20, no.~3, 2001.

\bibitem{physionet:00}
A.~L. Goldberger, L.~A.~N. Amaral, L.~Glass, J.~M. Hausdorff, R.~G.~M.
  P.~Ch.~Ivanov, J.~E. Mietus, G.~B. Moody, C.-K. Peng, and H.~E. Stanley,
  ``{P}hysio{B}ank, {P}hysio{T}oolkit, and {P}hysio{N}et {C}omponents of a
  {N}ew {R}esearch {R}esource for {C}omplex {P}hysiologic {S}ignals,''
  \emph{Circulation}, vol. 101, no.~23, 2000.

\bibitem{agrawal2018rewriting}
A.~Agrawal, R.~Verschueren, S.~Diamond, and S.~Boyd, ``{A} {R}ewriting {S}ystem
  for {C}onvex {O}ptimization {P}roblems,'' \emph{Journal of Control and
  Decision}, vol.~5, no.~1, pp. 42--60, 2018.

\bibitem{diamond2016cvxpy}
S.~Diamond and S.~Boyd, ``{CVXPY}: {A} {P}ython-embedded modeling language for
  convex optimization,'' \emph{Journal of Machine Learning Research}, vol.~17,
  no.~83, pp. 1--5, 2016.

\bibitem{CVXOPT}
M.~Andersen, J.~Dahl, and L.~Vandenberghe, ``{CVXOPT}: {P}ython {S}oftware for
  {C}onvex {O}ptimization.''

\bibitem{SSS:11}
R.~Sanyal, F.~Sottile, and B.~Sturmfels, ``Orbitopes,'' \emph{Mathematika},
  vol.~57, no.~2, p. 275–314, 2011.

\bibitem{SanSau:20}
R.~Sanyal and J.~Saunderson, ``{S}pectral {P}olyhedra,'' \emph{CoRR}, vol.
  abs/2001.04361, 2020.

\bibitem{FSP:15}
H.~Fawzi, J.~Saunderson, and P.~A. Parrilo, ``{E}quivariant {S}emidefinite
  {L}ifts and {S}um-of-{S}quares {H}ierarchies,'' \emph{SIAM Journal on
  Optimization}, vol.~25, no.~4, p. 2212–2243, 2015.

\end{thebibliography}

\end{document}